\documentclass[12pt]{article}

\usepackage{graphicx}
\usepackage{amsmath}
\usepackage{amssymb}
\usepackage{amsthm}
\usepackage{algpseudocode}
\usepackage{algorithm}
\usepackage{verbatim}
\usepackage{fancyvrb}
\usepackage{subfig}
\usepackage{url}
\usepackage{color}
\usepackage{cancel}
\usepackage{hyperref}
\usepackage{authblk}
\usepackage{listings}
\usepackage{ascii}
\usepackage{diagbox}
\usepackage{anysize}
\usepackage{caption}
\usepackage{adjustbox,expl3,etoolbox}

\usepackage{xy}
\xyoption{matrix}
\xyoption{frame}
\xyoption{arrow}
\xyoption{arc}

\usepackage{ifpdf}
\ifpdf
\else
\PackageWarningNoLine{Qcircuit}{Qcircuit is loading in Postscript mode.  The Xy-pic options ps and dvips will be loaded.  If you wish to use other Postscript drivers for Xy-pic, you must modify the code in Qcircuit.tex}
\xyoption{ps}
\xyoption{dvips}
\fi

\entrymodifiers={!C\entrybox}

\newcommand{\qw}[1][-1]{\ar @{-} [0,#1]}
\newcommand{\qwx}[1][-1]{\ar @{-} [#1,0]}

\newcommand{\gate}[1]{*+<.6em>{#1} \POS ="i","i"+UR;"i"+UL **\dir{-};"i"+DL **\dir{-};"i"+DR **\dir{-};"i"+UR **\dir{-},"i" \qw}

\newcommand{\control}{*!<0em,.025em>-=-<.2em>{\bullet}}

\newcommand{\ctrl}[1]{\control \qwx[#1] \qw}

\newcommand{\targ}{*+<.02em,.02em>{\xy ="i","i"-<.39em,0em>;"i"+<.39em,0em> **\dir{-}, "i"-<0em,.39em>;"i"+<0em,.39em> **\dir{-},"i"*\xycircle<.4em>{} \endxy} \qw}
\newcommand{\qswap}{*=<0em>{\times} \qw}
\newcommand{\Qcircuit}{\xymatrix @*=<0em>}

\letcs\replicate{prg_replicate:nn}
\newcommand*\longsum[1][1]{%
  \mathop{\textnormal{%
    \clipbox{0pt 0pt {.5\width} 0pt}{$\displaystyle\sum$}%
    \replicate{#1}{\clipbox{{.5\width} 0pt {.4\width} 0pt}{$\displaystyle\sum$}}%
    \clipbox{{.6\width} 0pt 0pt 0pt}{$\displaystyle\sum$}}}%
}

\marginsize{0.3in}{0.3in}{0.3in}{0.5in}

\newcommand{\be}{\begin{equation}}
\newcommand{\ee}{\end{equation}}
\newcommand{\ba}{\begin{array}}
\newcommand{\ea}{\end{array}}
\newcommand{\bea}{\begin{eqnarray}}
\newcommand{\eea}{\end{eqnarray}}

\newcommand{\F}{\mathbb{F}}

\newcommand{\ZZ}{\mathbb{Z}}

\newcommand{\RR}{\mathbb{R}}

\newcommand{\Gate}[1]{\textsc{#1}}
\newcommand{\hgate}{\Gate{h}}

\newcommand{\czgate}{\Gate{cz}}

\newcommand{\pgate}{\Gate{p}}
\newcommand{\idgate}{\Gate{i}}

\newcommand{\cnotgate}{\Gate{cnot}}

\newtheorem{lemma}{Lemma}

\newtheorem{corol}{Corollary}

\renewcommand{\sec}[1]{\hyperref[sec:#1]{Section~\ref*{sec:#1}}}
\newcommand{\ssec}[1]{\hyperref[ssec:#1]{Subsection~\ref*{ssec:#1}}}
\newcommand{\appe}[1]{\hyperref[appe:#1]{Appendix~\ref*{appe:#1}}}
\newcommand{\fig}[1]{\hyperref[fig:#1]{Fig.~\ref*{fig:#1}}}
\newcommand{\tab}[1]{\hyperref[tab:#1]{Table~\ref*{tab:#1}}}
\newcommand{\lem}[1]{\hyperref[lem:#1]{Lemma~\ref*{lem:#1}}}
\newcommand{\cor}[1]{\hyperref[cor:#1]{Corollary~\ref*{cor:#1}}}
\newcommand{\thm}[1]{\hyperref[thm:#1]{Theorem~\ref*{thm:#1}}}

\def\Sp{{\mathrm Sp}(2n,\F_2)}
\newcommand{\loc}{\mathsf{loc}}

\newcommand{\C}{{\cal C}}
\newcommand{\CU}{\mathfrak{C}}
\newcommand{\R}{{\cal R}}
\renewcommand{\S}{{\cal S}}
\newcommand{\D}{{\cal D}}
\newcommand{\PL}{{\cal P}}

\newcommand{\laptophottime}{0.0009358} 
\newcommand{\serverhottime}{0.0006274} 

\begin{document}

\title{6-qubit Optimal Clifford Circuits}
\author[1]{Sergey Bravyi}
\author[2]{Joseph A. Latone}
\author[1]{Dmitri Maslov}
\affil[1]{\small{IBM Quantum, IBM T. J. Watson Research Center, Yorktown Heights, NY 10598, USA}}
\affil[2]{\small{IBM Quantum, Almaden Research Center, San Jose, CA 95120, USA}}

\maketitle

\abstract{
Clifford group lies at the core of quantum computation---it underlies quantum error correction, its elements can be used to perform magic state distillation and they form randomized benchmarking protocols, Clifford group is used to study quantum entanglement, and more.  The ability to utilize Clifford group elements in practice relies heavily on the efficiency of their circuit-level implementation.  Finding short circuits is a hard problem; despite Clifford group being finite, its size grows quickly with the number of qubits $n$, limiting known optimal implementations to $n{=}4$ qubits.  For $n{=}6$, the number of Clifford group elements is about $2.1{\cdot}10^{23}$.  In this paper, we report a set of algorithms, along with their C\texttt{++} implementation, that implicitly synthesize optimal circuits for all $6$-qubit Clifford group elements by storing a subset of the latter in a database of size $2.1$TB ($1$KB${=}1024$B).  We demonstrate how to extract arbitrary optimal $6$-qubit Clifford circuit in $\laptophottime$ and $\serverhottime$ seconds using consumer- and enterprise-grade computers (hardware) respectively, while relying on this database.  We use this implementation to establish a new example of quantum advantage by Clifford circuits over $\cnotgate$ gate circuits and find optimal Clifford 2-designs for up to $4$ qubits.
}

\section{Introduction}
Quantum computations are studied for their promise to outperform classical counterparts for certain kinds of computations \cite{nielsen2002quantum}.  Clifford group is an important finite subgroup of the full unitary group, describing the set of quantum computations.  Despite being possible to simulate classically \cite{gottesman1998heisenberg, aaronson2004improved} by a low degree polynomial and having a simple structure \cite{bravyi2020hadamard} (admitting efficient parametrization and being possible to compute by linear depth circuits), the group is most famous for lying at the core of quantum error correction \cite{nielsen2002quantum}, which is believed to be necessary for scalable quantum computation.  Restricted to the study of fault-tolerance, Clifford group plays multiple roles still. To illustrate, all (standard) encoding circuits are Clifford \cite{nielsen2002quantum}, and so are the circuits for state distillation \cite{bravyi2005universal, knill2005quantum}, necessary for fault-tolerant implementation of non-Clifford gates.  Clifford circuits lie at the core of randomized benchmarking protocols \cite{knill2008randomized, magesan2011scalable}.  Other use cases include shadow tomography \cite{aaronson2020shadow, huang2020predicting}, study of entanglement \cite{nielsen2002quantum, bennett1996mixed}, and quantum data hiding \cite{divincenzo2002quantum}.  
It is perhaps fair to regard the Clifford group as one of the most visible and important subgroups of the group of all quantum computations.

Superconducting circuits and trapped ions are two technological frameworks that produced a stream of (universal prototype) programmable quantum computers, publicly available since the year 2016.  Each technology comes in a range of flavors: e.g., superconducting circuits can be based on phase, charge, or flux qubits (or even hybrid kinds), and rely on various qubit coupling mechanisms, and trapped ions can be based on various ion species and rely on different approaches to the two-qubit gates (e.g., stationary vs mobile qubits).  However, no matter the specific flavor, all prototype quantum computers based on these two approaches share one property \cite{IBM, AWS}: the two-qubit gate has notably lower fidelity than a single-qubit gate.  Thus, to the first degree of approximation, the fidelity of an entire quantum computation depends on the number of two-qubit gates it uses.  To make a more subtle point, since the single-qubit gates are most frequently implemented by pulses with real-valued control parameters, the number of two-qubit gates in a circuit upper bounds the number of the single-qubit gates (up to a constant factor), meaning the reduction of the two-qubit gate count likely leads to the reduction in the number of single-qubit gates.  We further note that the $\cnotgate$ gates are available natively (i.e., requiring the minimal number of one two-qubit physical-level interaction) in both superconducting circuits and trapped ions technologies.  Finally, recall that the physical-level entangling pulses frequently take the form of $XX$, $ZX$, and $ZZ$, requiring single-qubit corrections to turn those interactions into commonly used $\cnotgate$ or $\czgate$ gates.  This means that minimizing single-qubit gate count in an abstract circuit may not directly minimize the number of single-qubit physical pulses, since the single-qubit gates will be reshuffled during technology mapping.   This justifies our focus on minimizing the $\cnotgate$ gate count, selected as the optimization criterion in this paper. 

In this paper, we study the problem of optimal synthesis of Clifford circuits.  Since the problem of optimal circuit synthesis is hard, we restrict our attention to a small number of qubits, at most $6$.  The number of Clifford group elements over $6$ qubits, $2.1{\cdot}10^{23}$, is still very large, and we employ a range of techniques to make the search tractable using modern computers.  At the core of our approach is a mechanism to break down the set of Clifford unitaries into a set of classes containing unitaries sharing a similar optimal circuit structure, efficient computation of the canonical representative of each class, and efficient manipulation of class members and the database of canonical representatives. 

The rest of the paper is organized as follows. \sec{def} reports definitions necessary to understand the technical parts. \sec{alg} starts with a subsection containing an overview of our algorithm; all technical details can be found in the following five subsections.  \sec{results} discusses the results, including a summary of relevant statistics (average and maximal circuit sizes, distribution of optimal costs), properties of optimal Clifford circuits that were possible to calculate using the data synthesized (advantage of Clifford circuits over linear reversible circuits, optimal 2-designs), and compares our work to previous similar results.

\section{Definitions} 
\label{sec:def}
We define the $n$-qubit Clifford group $\C_n$ as the group of $2n {\times} 2n$ symplectic matrices $M$ over the two-element field $\F_2$, $\Sp\,{:=}\,\{M{:}\; M^T\Omega_nM \,{=}\, \Omega_n\}$, where $M^T$ denotes transpose matrix, $\Omega_n$ is the matrix $\begin{pmatrix} 0 & I_n \\ I_n & 0 \end{pmatrix}$, and $I_n$ is the $n{\times}n$ identity matrix.  Symplectic matrices are equivalent to and alternatively known as the tableaux \cite{aaronson2004improved}.  The size of the symplectic group is $|\Sp|={2^{n^2}\prod\limits_{j=1}^{n}(2^{2j}-1)}$, which for the purpose of this paper implies $|\C_6|\,{=}\, 208,\!114,\!637,\!736,\!580,\!743,\!168,\!000 \,{\approx}\, 2.1{\cdot}10^{23}$ and assigns the numeric value to the size of the search space we are exploring.

Tableau representation is particularly useful since it allows to define quantum gates and circuits directly without the need to resort to standard definitions in quantum information that employ $2^n{\times}2^n$ unitary matrices \cite{nielsen2002quantum}.  Indeed, 
\begin{itemize}
    \item the Hadamard gate $\hgate$ on qubit $k$ can be defined as the $2n{\times}2n$ identity matrix with swapped columns $k$ and $n{+}k$,
    \item the Phase gate $\pgate$ on qubit $k$ can be defined as the addition of column $k$ to column $n{+}k$ in the $2n{\times}2n$ identity matrix,
    \item the $\cnotgate$ gate with control qubit $k$ and target $j$ performs simultaneous addition of column $k$ to column $j$ and 
    column $n{+}j$ to column $n{+}k$ in the $2n{\times}2n$ identity matrix,
\end{itemize}
\noindent and circuits are matrix multiplications.  The computational completeness of the $\{\hgate,\pgate,\cnotgate\}$ library is readily exposed by the ability to apply Gaussian elimination to obtain arbitrary symplectic matrix as a product of gates.  An additional advantage of such a definition of gates and circuits comes from displaying the capacity to implement transformations by Clifford gates efficiently by a computer program.

As a side note, we highlight that each element of the Clifford group $\C_n$ defines an equivalence class
of $2^n {\times} 2^n$ unitary matrices realizable by the circuits over $\hgate$, $\pgate$, and $\cnotgate$ gates (defined, in turn, via unitary matrices \cite{nielsen2002quantum}). A pair of unitary matrices is considered equivalent if they can be mapped to each other by the left (or right) multiplication with single-qubit Pauli gates and overall phase factors.
Since we focus on the  minimization of the two-qubit gate count, Pauli gates and phase factors can be safely factored out.
Had Pauli gates been included in the Clifford group, the search space size for $n{=}6$ would read $8.5{\cdot}10^{26}$.

\section{Algorithm and its implementation}\label{sec:alg}
\subsection{Overview}
Our approach relies on the use of pruned breadth-first search (BFS) to generate a number of databases containing Clifford unitaries that can be implemented by equal cost optimal circuits, and augment it by a set of tools that extract useful statistics (e.g., distribution of the number of unitaries by entangling gate cost, average cost, largest cost) as well as individual optimal circuits.  BFS is a strategy that relies on taking optimal implementations of cost up to $k$, modifying them by applying cost-$1$ transformations to cost-$k$ elements, and recording the result as a cost $k{+}1$ element if it is not yet found in the database.  BFS is initiated with the identity operator costing zero and ends when all elements in the target set were explored.  While our algorithm can be applied to obtain optimal $2$-, $3$-, $4$-, $5$-, and $6$-qubit Clifford circuits using modern computers, we focus the rest of the description on the most difficult but still amenable to classical computers $6$-qubit case. 

Since the database we are synthesizing contains Clifford unitaries, the first order of business is to choose a suitable data structure to store those.  The data structure must be both compact and allow quick application of gates; this is because BFS boils down to a series of gate applications and memory lookups.  We start with the tableau, which is naturally suited for quick gate application, and modify it to remove two last rows corresponding to $X$ and $Z$ stabilizers each \cite{aaronson2004improved}.  As described in \ssec{datastructure}, these rows can be quickly restored.  However, removing them allows to reduce the storage from $4n^2|_{n=6}=144$ bits to $2\cdot 2n(n{-}1)|_{n=6} = 120$ bits.  Each unitary is thus stored across two $64$-bit machine words (each half corresponding to $X$ and $Z$ parts), with $4$ bits per machine word of (yet) unused space.  While information-theoretic minimum storage requirement, $\lceil\log_2(|\C_6|)\rceil \,{=}\, 78$, implies that more compact storage exists, BFS imposes the requirement of quick gate application and we furthermore rely on canonicity (discussed in next paragraph) to reduce the size of the database; thus, it is not obvious if more efficient storage is possible.

Should each Clifford element require storage, the search would not be possible to execute on modern computers since $|\C_6|\,{\approx}\, 2{\cdot}10^{23}$.  We, therefore, break Clifford group elements into classes of equivalence such that class members share the same optimal circuit structure, a canonical representative exists, and it is efficient to compute.  In our approach, a class of equivalence can be thought of as containing unitaries with optimal circuits equivalent up to left- and right-hand multiplication by single-qubit Clifford unitaries, and qubit relabeling; the canonical representative is chosen to be the one with the least lexicographic order across all elements in its equivalence class.  This means that we can pack up to $|\C_1|^{2n}{\cdot} |S_n|\Big|_{n=6} = 6^{12} \cdot 6! = 1{,}567{,}283{,}281{,}920$ unitaries into one class\footnote{
More precisely, the number of unitaries contained in each equivalence class may vary between
$|\C_1|^n$ and 
$|\C_1|^{2n}{\cdot} |S_n|$. 
The former case is realized for the identity operator which is invariant under all qubit relabelings and does not differentiate between left- and right-hand multiplications by single-qubit Clifford unitaries. The latter case is realized for a generic element of the Clifford group without any special symmetries.}.
Here, $|\C_1|$ is the size of the single-qubit Clifford group $\C_1$ raised to the power $2n$ to represent one-qubit operators on each qubit in the beginning and end of the circuit, and $S_n$ is the permutation group.  However, the computation of canonical representative must be efficient, as otherwise, complexity moves from storage to computation.  We utilized a Pareto-efficient definition of the equivalence class, as determined by $\mathsf{ReduceU}$, the function computing the canonical representative, to be most practical.  Our computationally-defined canonical representative is at most factor $14$ storage inefficient, but it allows a quick computation of the canonical representative, taking on average $0.000003$ seconds (using Intel Core i7-10700K processor).  The computation of $\mathsf{ReduceU}$ turns out to be the runtime-level bottleneck of our implementation since other operations that are applied with a comparable frequency (such as tableau restoration and gate application) are faster.  Further details about $\mathsf{ReduceU}$ may be found in \ssec{reduceu}.

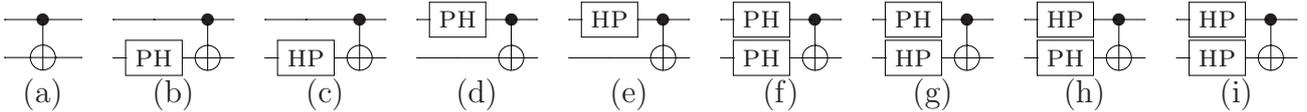
\begin{figure}[t]
\centering
\begin{center}
\begin{tabular}{ccccccccc}
\Qcircuit @C=0.4em @R=0.4em @! {
& \ctrl{1}  & \qw \\
& \targ     & \qw
}
&
\Qcircuit @C=0.4em @R=0.1em @!R {
& \qw                   & \ctrl{1} & \qw \\
& \gate{\pgate\hgate}   & \targ & \qw 
}
&
\Qcircuit @C=0.4em @R=0.1em @!R {
& \qw                   & \ctrl{1} & \qw \\
& \gate{\hgate\pgate}   & \targ & \qw 
}
&
\Qcircuit @C=0.4em @R=0.1em @!R {
& \gate{\pgate\hgate}   & \ctrl{1} & \qw \\
& \qw                   & \targ & \qw 
}
&
\Qcircuit @C=0.4em @R=0.1em @!R {
& \gate{\hgate\pgate}   & \ctrl{1} & \qw \\
& \qw                   & \targ & \qw 
}
&
\Qcircuit @C=0.4em @R=0.1em @!R {
& \gate{\pgate\hgate}   & \ctrl{1} & \qw \\
& \gate{\pgate\hgate}   & \targ & \qw 
}
&
\Qcircuit @C=0.4em @R=0.1em @!R {
& \gate{\pgate\hgate}   & \ctrl{1} & \qw \\
& \gate{\hgate\pgate}   & \targ & \qw 
}
&
\Qcircuit @C=0.4em @R=0.1em @!R {
& \gate{\hgate\pgate}   & \ctrl{1} & \qw \\
& \gate{\pgate\hgate}   & \targ & \qw 
}
&
\Qcircuit @C=0.4em @R=0.1em @!R {
& \gate{\hgate\pgate}   & \ctrl{1} & \qw \\
& \gate{\hgate\pgate}   & \targ & \qw 
}
\\
(a) & (b) & (c) & (d) & (e) & (f) & (g) & (h) & (i)
\end{tabular}
\end{center}
\begin{minipage}[c]{0.85\linewidth}
\caption{$\cnotgate$ gate equivalent entangling transformations that need to be applied to each of $\frac{n(n{-}1)}{2}$ pairs of qubits of a Clifford group element implementable with $k$ entangling gates to explore the possibility of expanding it into a Clifford group element requiring $k{+}1$ gates.  It suffices to apply these gates to a pair of qubits in an arbitrary fixed order, since the application of a gate in the other order is enabled by some other gate among those listed.  For instance, the $\cnotgate$ with flipped controls with respect to (a) is accomplished by (h), noting that the single-qubit gates on the right side do not matter due to the choice to work with equivalence classes.}
\label{fig:135}
\end{minipage}
\end{figure}

The restriction to equivalence classes helps not only to dramatically reduce the storage requirement, but also to minimize the number of $\cnotgate$-equivalent transformations that we need to apply to a Clifford unitary requiring $k$ gates to explore Clifford unitaries requiring $k{+}1$ entangling gates.  Specifically, the number of transformations is only  $9\frac{n(n{-}1)}{2}\big|_{n=6} = 135$, as illustrated in \fig{135}.

The $15$-part (one part per a fixed gate count ranging from $1$ to $15$, with $15$ turning out to be the maximum) sorted database with canonical representatives of equal cost is $2.1$TB in size, and it took roughly $6$ months to synthesize it on a small cluster of Intel\textsuperscript{\textregistered} server-class machines.  Since we made software updates as the search progressed, and improved the performance in doing so, we believe it may take about $2$ months to rerun it from scratch.  We store the database on an SSD ($2+$TB RAM was expensive at the time of this writing).  Given the database, an optimal circuit for a given $6$-qubit Clifford unitary $U$ may be found as follows: compute $\mathsf{ReduceU}(U)$, find it in part of the database containing size $k$ unitaries, apply each of $9\frac{n(n{-}1)}{2}$ gates, compute the resulting canonical element and look it up in the size $k{-}1$ database; once found repeat for $k\,{:=}\,k{-}1$ until $k{=}0$.  Our implementation of the above algorithm takes an average of $0.1$ seconds to extract an optimal circuit.  The bottleneck is the database search on the SSD, since the average number of times an element needs to be searched is at most $\frac{135}{2} = 67.5$, the databases for large $k$ are large, and search needs to make multiple queries that add up quickly given SSD's limited access time.  Instead, recall that $4{+}4\,{=}\,8$ bits of the original data structure are unused, and note that $8$ bits suffice to store the gate information, since $\lceil\log_2(135)\rceil \,{=}\, 8$.  We thus augment the database by loading these $8$ bits with the last gate information, allowing to select the correct gate right away during the circuit restoration.  This modification reduces the runtime by roughly a factor of $67.5$.  We further optimize the performance by storing an index with each $1024^{\text th}$ element of the database in RAM.  This allows finding an optimal circuit implementation of an arbitrary $6$-qubit Clifford unitary in as little as $\laptophottime$ seconds on a MacBook Pro\textsuperscript{\textregistered} (2.3 GHz Quad-Core Intel\textsuperscript{\textregistered} Core i7-1068NG7 CPU, 16GB RAM) with a USB-C attached SSD (4TB VectoTech Rapid\textsuperscript{\textregistered} 540MB/s 3D NAND Flash), and $\serverhottime$ seconds on a high-performance server (Quad Intel\textsuperscript{\textregistered} Xeon E7-4850 v4 16-Core/2.1GHz, 6TB RAM).  These performance figures were established by averaging out the time to synthesize optimal circuits for $10{,}000$ random uniformly distributed Clifford unitaries while relying on kernel-owned memory to cache files with the use of \textit{mmap} and using a supplementary index for the laptop version of the search. 

In the following subsections we report further details of our implementation.

\subsection{Database generation}
\label{ssec:database}

Let $\C_n^k\,{\subseteq}\, \C_n$ be the set of all Clifford group elements with the $\cnotgate$ cost $k$. Here $k=0,1,\ldots,k_{max}(n)$ for some a-priori unknown maximum cost $k_{max}(n)$. For example, $\C_n^0$ is the local subgroup of $\C_n$, i.e., one generated by the single-qubit Clifford gates.  Suppose $\mathsf{ReduceU}{:}\; \C_n\to \C_n$ is a function such that $\mathsf{ReduceU}(U)\,{=}\,\mathsf{ReduceU}(V)$ if and only if $U$ and $V$ are equivalent up to left and right multiplications by single-qubit gates and a qubit relabeling.  In other words, $\mathsf{ReduceU}(U)$ is a canonical representative of the equivalence class 
\be
\label{eq_class}
[U]:=\{KW^{-1}UWL\,{:} \; K,L\in \C_n^0,\; W\in S_n\}.
\ee
Here and below $S_n\,{\subseteq}\, \C_n$ is the subgroup of qubit permutations.  A specific implementation of the function $\mathsf{ReduceU}$, which we defer to \ssec{reduceu}, does not matter at this point.  Let $\R_n^k$ be the set of all reduced cost-$k$  Clifford group elements,
\[
\R_n^k := \{ \mathsf{ReduceU}(U)\,{:} \; U\in \C_n^k\}.
\]
Our database consists of $k_{max}(n){+}1$ parts, such that the $k$-th part contains all elements of $\R_n^k$.  The elements are furthermore stored in the lexicographic order to enable binary search. 

Let $I\,{\in}\,\C_n$ be the identity matrix and $\cnotgate_{i,j}$ be the $\cnotgate$ gate with the control qubit $i$ and the target qubit $j$.  Since any cost-$0$ and cost-$1$ element is equivalent to $I$ and $\cnotgate_{1,2}$ respectively, we have
\[
\R_n^0= \{\mathsf{ReduceU}(I)\} \mbox{ \;and\; } \R_n^1=\{ \mathsf{ReduceU}(\cnotgate_{1,2})\}.
\]

Suppose we have the sets $\R_n^0, \R_n^1,\ldots, \R_n^{k-1}$ for some $k\,{\ge}\,2$ (initially $k{=}2$).  The rest of this section explains how to compute $\R_n^k$.  First, we need to choose a set of cost-$1$ generators that obey certain technical conditions. Let $m=9n(n{-}1)/2$ and $G_1,G_2,\ldots,G_m \in \C_n^1$ be the generators shown in \fig{135}.  By definition, each generator has the form $\Gate{a}_i \Gate{b}_j \cnotgate_{i,j}$ for some pair of qubits $i{<}j$ and $\Gate{a},\Gate{b}\in \{\idgate,\pgate\hgate,\hgate\pgate \}$.  We will use the following properties of the generator set.
\begin{lemma}
\label{lem:gen1}
Any cost-$k$ element $U\,{\in}\, \C_n^k$ can be written as $U=G_{a_1}G_{a_2}\cdots G_{a_k}L$ for some $L\,{\in}\, \C_n^0$ and some $a_1,a_2,\ldots,a_k\in \{1,2,\ldots,m\}$.
\end{lemma} 
The proof is deferred to \appe{A}.  This lemma has the following simple corollaries.
\begin{corol}
\label{cor:corol2}
Suppose $W\,{\in}\, S_n$ is a qubit permutation and $L\,{\in}\, \C_n^0$. For any generator $G_a$ there exist a generator $G_b$ and $M\,{\in}\, \C_n^0$ such that $WLG_a=G_b WM$.
\end{corol}
\begin{proof}
Let $U\,{=}\,WLG_a W^{-1}$.  Note that $U\,{\in}\, \C_n^1$ since $U$ is equivalent to a cost-$1$ element $G_a$.  \lem{gen1} with $k{=}1$ implies that $U\,{=}\,G_b M'$ for some generator $G_b$ and some $M'\,{\in}\, \C_n^0$.  Thus $WLG_a=G_b M' W=G_bW M$, where $M=W^{-1} M' W\in \C_n^0$. 
\end{proof}
\begin{corol}
\label{cor:corol1}
For any generator $G_a$ and $L\,{\in}\, \C_n^0$ there exists a generator $G_b$ such that $G_a L G_b \,{\in}\, \C_n^0$.
\end{corol}
\begin{proof}
Let $U\,{=}\,(G_aL)^{-1}$. Note that $U\,{\in}\, \C_n^1$ since the cost is invariant under taking the inverse. \lem{gen1} with $k{=}1$ implies that $U\,{=}\,G_b M$ for some generator $G_b$ and $M\,{\in}\, \C_n^0$. Thus $G_a L G_b=M^{-1}\in \C_n^0$.
\end{proof}
We claim that the following algorithm outputs the set $S\,{=}\,\R_n^k$.

\begin{center}
\fbox{\parbox{7cm}{
\begin{algorithmic}
\State $S\gets \emptyset$
\For{$V\in \R_n^{k-1}$}
\For{$b\in \{1,2,....,m\}$}
\State $U\gets \mathsf{ReduceU}(VG_b)$
\If{$U\notin \R_n^{k-2} \cup \R_n^{k-1}$}
\State $S \gets S \cup \{U\}$.
\EndIf 
\EndFor
\EndFor
\end{algorithmic}}}
\end{center}

Let us first check that $\R_n^k\,{\subseteq}\, S$.  Consider any element $U\,{\in}\, \R_n^k$.  Then $U\,{=}\,\mathsf{ReduceU}(\tilde{U})$ for some $\tilde{U}\,{\in}\, \C_n^k$.  By \lem{gen1}, we can write $\tilde{U}\,{=}\,G_{a_1}G_{a_2}\cdots G_{a_k}M$ for some $M\,{\in}\, \C_n^0$.  Define
\[
\tilde{V}:=G_{a_1}G_{a_2}\cdots G_{a_{k-1}} \mbox{ \;and\; } V:= \mathsf{ReduceU}(\tilde{V}).
\]
Note that $\tilde{V}\,{\in}\, \C_n^{k-1}$ (if $\tilde{V}\in \C_n^\ell$ for some $\ell\,{<}\,k{-}1$ then $\tilde{U}\,{=}\,\tilde{V}G_{a_k}M$ would have cost less than $k$).  Accordingly, $V\,{\in}\, \R_n^{k-1}$.  By definition of the function $\mathsf{ReduceU}$, we have $\tilde{V}=KW^{-1} V WL$ for some $K,L\in \C_n^0$ and some qubit relabeling $W\,{\in}\, S_n$. Thus 
\[
\tilde{U}=G_{a_1}G_{a_2}\cdots G_{a_k}M=\tilde{V}G_{a_k}M=KW^{-1} V WLG_{a_k}M.
\]
Commuting $G_{a_k}$ through $WL$ next to $V$ using \cor{corol2} we obtain  $\tilde{U}=KW^{-1} (VG_b) WM'$ for some generator $G_b$ and some $M'\,{\in}\, \C_n^0$. This shows that $\tilde{U}$ is equivalent to $VG_n$ and thus $\mathsf{Reduce}(VG_b)=\mathsf{Reduce}(\tilde{U})=U$ for some $V\in \R_n^{k-1}$ and some generator $G_b$. Thus $U\,{\in}\, S$.  We have proved that $\R_n^k\,{\subseteq}\, S$.

Conversely, suppose $U\,{\in}\, S$.  Then $U$ is a reduced element obtained from some cost-$(k{-}1)$ element $V$ by adding a single generator, relabeling the qubits, and left/right multiplications by the single-qubit gates.  Since adding a single generator can change the cost by at most one\footnote{The cost cannot grow by more than 1 for an obvious reason.  It cannot decline by $d{>}1$ since this would imply that $V$ can be implemented with cost $(k{-}1{-}d)+1 = k-d < k{-}1$ as the circuit $(Vg).g^{-1}$, where $g$ is the generator, which contradicts the notion that $V$ is a cost-$(k{-}1)$ element.}, we conclude that $U\in \R_n^{k-2} \cup \R_n^{k-1}\cup \R_n^k$.  Thus the algorithm adds $U$ to $S$ only if $U\,{\in}\, \R_n^k$.  We have proved that $S\,{\subseteq}\, \R_n^k$.

By sorting the elements of each set $\R_n^\ell$ and using the binary search to check set membership, the above algorithm requires $\tilde{O}(|\R_n^{k-1}|m)$  calls to the function $\mathsf{ReduceU}$, where the $\tilde{O}$ notation hides factors logarithmic in the size of $\R_n^{k-2}$, $\R_n^{k-1}$, and $\R_n^k$.  The database generation terminates as soon as $\R_n^k\,{=}\,\emptyset$.  This determines the maximum cost $k_{max}(n)$ as $k{-}1$.

As discussed in \sec{alg}, the generation of the $6$-qubit database spans a few CPU months and involves manipulations with terabytes of data.  How can we be confident that this computation is error-free?  Our correctness tests included the verification that the size of the Clifford group inferred from the database agrees with the analytic formula $|\C_n|={2^{n^2}\prod_{j=1}^{n}(4^j{-}1)}$.  In more detail, the number of cost-$k$ Clifford group elements can be inferred from the identity
\be
\label{count_test1}
|\C_n^k| = \sum_{U\in \R_n^k} |[U]|,
\ee
where $|[U]|$ is the size of the equivalence class $[U]$ that contains $U$, see Eq.~(\ref{eq_class}). 
Furthermore, 
\be
\label{count_test2}
|[U]| =  \frac{|\C_n^0|^2 \cdot |S_n|}{|\mathrm{Aut}(U)|} =  \frac{6^{2n} n!}{|\mathrm{Aut}(U)|},
\ee
where $\mathrm{Aut}(U)$ is the automorphism group of $U$ that consists of all triples $K{\times} L{\times} W\in \C_n^0{\times} \C_n^0{\times} S_n$ such that  $U=KW^{-1}UWL$. We have checked that the counts $|\C_n^k|$ inferred from Eqs.~(\ref{count_test1},\ref{count_test2}) indeed obey $\sum_{k=0}^{k_{max}(n)} |\C_n^k|=|\C_n|$. Thus our database passed the self-consistency test.  \tab{2345} and \tab{6qdistribution} displaying the counts $|\R_n^k|$ and $|\C_n^k|$ can be found in \sec{results}.

In order to speed up the synthesis of optimal circuits, we augmented each database entry $U\,{\in}\, \R_n^k$  with $8$ auxiliary bits specifying a generator $G_b$ that reduces the cost of $U$ by one, such that $UG_b\,{\in}\, \C_n^{k-1}$.  Here we assume $k{\ge} 1$.  Let us prove that such cost-reducing generator $G_b$ exists for any $U\,{\in}\, \R_n^k$.  Indeed, use \lem{gen1} to write $U=G_{a_1} G_{a_2} \cdots G_{a_k} L$ for some $L\,{\in}\, \C_n^0$. By \cor{corol1}, there exists a generator $G_b$ such that $F\equiv G_{a_k}L G_b\in \C_n^0$.  Now $UG_b=G_{a_1} G_{a_2} \cdots G_{a_{k-1}}F$ for some $F\,{\in}\, \C_n^0$, that is, $UG_b$ has cost $k{-}1$. 

To augment a given element $U$ of the cost-$k$ database $\R_n^k$ we find the first  cost-reducing generator $b\,{\in}\, \{1,2,\ldots,m\}$ such that $\mathsf{ReduceU}(UG_b)\,{\in}\, \R_n^{k-1}$.  This requires at most $m$ calls to $\mathsf{ReduceU}$ and binary searches in $\R_n^{k-1}$ (computing the group multiplication takes a negligible time).  Once a cost-reducing generator $G_b$ is found, its index $b$ is recorded in the database using the unused bits of $U$. The augmentation step is applied to all $U\,{\in}\, \R_n^k$ and for all $k\,{=}\,1,2,\ldots,k_{max}(n)$.

\subsection{Synthesis of optimal circuits}
\label{ssec:synthesis}

The optimal compiler takes as input an element of the Clifford group $U\,{\in}\, \C_n$ and outputs a Clifford circuit (a list of the primitive gates $\hgate$, $\pgate$, and $\cnotgate$) implementing $U$ with the smallest possible $\cnotgate$ gate count, equal to the cost of $U$.  The cost can be computed by making a single call to $\mathsf{ReduceU}$ and performing at most $k_{max}(n)$ database searches.  Below we assume that the database is augmented with the cost-reducing generators, as discussed in \ssec{database}.  Thus the database search returns the cost $k$ element $V$ such that $V\,{\equiv}\, \mathsf{Reduce}(U)\in \R_n^k$ and a cost-reducing generator $G_a$ such that $VG_a\,{\in}\, \C_n^{k-1}$.  The next step is to convert $G_a$ into a cost-reducing generator for $U$.  To this end, write $V\,{=}\,K W^{-1} U W L$ for some $K,L\in \C_n^0$ and some qubit permutation $W$.  The group elements $K$, $L$, and $W$ that transform $U$ into the reduced form are readily available by adding appropriate bookkeeping steps to the implementation of $\mathsf{ReduceU}$ described in \ssec{reduceu}. At this point we have 
\[
KW^{-1} U W L G_a \in \C_n^{k-1}.
\]
Commute $G_a$ through $WL$ next to $U$ using \cor{corol2}.  This gives
\[
KW^{-1} U G_b W M \in \C_n^{k-1}
\]
for some generator $G_b$ and some $M\,{\in}\, \C_n^0$.  The generator $G_b$ can be computed in time $O(1)$ using the standard commutation rules of the Clifford group. Thus $UG_b \,{\in}\, \C_n^{k-1}$, that is, $G_b$ is a cost-reducing generator for $U$.  Replacing $U$ by $UG_b$ and applying the above step recursively, one constructs a $k$-tuple of generators such that $M=UG_{a_1} G_{a_2} \cdots G_{a_k} \in \C_n^0$ is a product of single-qubit gates.  This gives $U^{-1} = G_{a_1} G_{a_2} \cdots G_{a_k}M^{-1}$.  Decomposing each generator and $M^{-1}$ into a product of primitive gates $\hgate$, $\pgate$, and $\cnotgate$  gives an optimal circuit implementing $U^{-1}$. Since all primitive gates are self-inverse, an optimal circuit implementing $U$ is obtained simply by reversing the order of gates.  If needed, the number of single-qubit gates in the compiled circuit can be optimized by commuting single-qubit gates to the last time step (whenever possible) and merging them using optimal lookup of $\C_1$ elements.

\subsection{Computation of ReduceU}\label{ssec:reduceu}

In this section we introduce reduced forms of Clifford group elements and give algorithms for computing these forms.  A given matrix $U\,{\in}\,\C_n$ is transformed into a reduced form by applying a sequence of elementary reductions from the following list:
\begin{enumerate}
\item Multiplication of $U$ on the left by single-qubit Clifford gates.
\item Multiplication of $U$ on the right by single-qubit Clifford gates.
\item Relabeling of qubits.
\end{enumerate}
Depending on which type of reductions is considered, there are three different reduced forms: a left-reduced form (reductions of type 1 only), a locally reduced form (reductions of types 1 and 2), and a fully reduced form (reductions of types 1, 2, and 3).  Each form comes with an algorithm specifying the sequence of reductions to be applied.  We define the reduced forms inductively starting from the left-reduced form.  The function $\mathsf{ReduceU}$ used in \ssec{database} and \ssec{synthesis} computes the fully reduced form.

We begin by defining convenient notations.  Let $e^1,e^2,\ldots,e^{2n}\in \F_2^{2n}$ be the standard basis of $\F_2^{2n}$: the basis vector $e^j$ has a single non-zero at the $j$-th position.  We consider $e^j$ as column vectors.  Let $e_j:= (e^j)^T$ be the corresponding row vector.  For example, if $n{=}1$ then 
\[
e^1=\left[ \ba{c} 1 \\ 0 \ea \right],
\quad
e^2=\left[ \ba{c} 0 \\ 1 \ea \right],
\quad
e_1=\left[ \ba{cc} 1 & 0 \\ \ea \right],
\mbox{ \;and\; }
e_2=\left[ \ba{cc} 0 & 1 \\ \ea \right].
\]
We write $u\,{\oplus}\, v$ to denote the addition of binary vectors $u$ and $v$ modulo $2$.  Elements of the Clifford group $U\,{\in}\,\C_n$ are treated as binary symplectic matrices of the size $2n{\times} 2n$, see \sec{def}.  A matrix $U$ has the $j$-th column and the $j$-th row $Ue^j$ and $e_j U$, respectively.

Recall that $\C_n^0\,{\subseteq}\, \C_n$ is the local subgroup generated by the single-qubit gates ($\hgate$ and $\pgate$).  Define a subgroup $\C_{n,j}\,{\subseteq}\, \C_n^0$ generated by the single-qubit gates acting on the $j$-th qubit, where $j=1,2,\ldots,n$.  Equivalently, $U\,{\in}\, \C_{n,j}$ iff $Ue^i=e^i$ for all $i\,{\notin}\, \{j,n+j\}$, whereas $Ue^j = ae^j \oplus b e^{n+j}$ and $Ue^{n+j} = c e^j \oplus d e^{n+j}$ for some coefficients $a,b,c,d\in \F_2$ such that 
\[
\left[ \ba{cc}
a & c \\
b & d \\
\ea \right] \in \mathrm{GL}(2,\F_2).
\]
Note that the subgroups $\C_{n,j}$ pairwise commute. 

A matrix $U\in \C_n$ is said to be {\em left-reduced} if
\begin{equation}
\label{left_reduced}
e_j U < e_{n+j} U < (e_j \oplus e_{n+j})U \mbox{ for all $j=1,2,\ldots,n$}.
\end{equation}
Here  and below the bit strings are compared using the lexicographic order (i.e., $00\,{<}\,01\,{<}\,10\,{<}\,11$ in the case $n{=}1$). The following lemma shows that left-reduced elements of $\C_n$ can serve as canonical representatives of cosets $\C_n^0U$.  In other words, $\C_n$ is a disjoint union of cosets $\C_n^0U$ and each coset contains a unique left-reduced element, which can be efficiently computed.  We refer to the unique left-reduced element of a coset $\C_n^0U$ as the {\em left-reduced form} of $U$ and denote it $\mathsf{leftReduce}(U)$.  Our symplectic matrix data structure described in \ssec{datastructure} enables the computation of $\mathsf{leftReduce}(U)$ for a randomly picked matrix $U\,{\in}\, \C_n$  in time less than $2{\cdot}10^{-8}$ seconds for any $n\,{\le}\,6$ on a server-class CPU, in this case an Intel\textsuperscript{\textregistered} Xeon\textsuperscript{\textregistered} CPU E7-4850 v4 @ 2.10GHz.

\begin{lemma}
\label{lem:left_reduced}
Each coset $\C_n^0U$ with $U\,{\in}\, \C_n$ contains a unique left-reduced element that can be computed in time $O(n^2)$, given symplectic matrix representation of $U$.
\end{lemma}
\begin{proof}
First note that the rows of a symplectic matrix are linearly independent.  Thus for each qubit $j$ the bit strings $x_j\,{:=}\, e_j U$, $z_j\,{:=}\, e_{n+j}U$, and $y_j\,{:=}\, (e_j {\oplus} e_{n+j})U$ are all distinct: $x_j\,{\ne}\, y_j\,{\ne}\, z_j$.  It follows directly from the above definitions that multiplying $U$ on the left by the elements of the subgroup $\C_{n,j}$ we can implement any permutation of the bit strings $x_j$, $y_j$, and $z_j$.  For example, the Hadamard gate swaps $x_j$ and $z_j$, the Phase gate swaps $x_j$ and $y_j$.  Since $|\C_{n,j}|\,{=}\,6$, there is a one-to-one correspondence between elements of $\C_{n,j}$ and permutations of $x_j$, $y_j$, $z_j$. Multiply $U$ on the left by the unique element of $\C_{n,j}$ that permutes the bit strings such that $x_j\,{<}\,z_j\,{<}\,y_j$. Now Eq.~(\ref{left_reduced}) is satisfied for the $j$-th qubit. Repeating this for all $n$ qubits and noting that $\C_n^0$ is generated by the subgroups $\C_{n,j}$ proves that the coset $\C_n^0U$ contains a unique left-reduced element.  All above steps can be efficiently implemented.  Indeed, given a matrix $U$, one can compute the bit strings $x_j$, $y_j$, and $z_j$ and sort all three in time $O(n)$. Repeating this for all $n$ qubits gives the total runtime of $O(n^2)$.
\end{proof}

Given a matrix $U\,{\in}\,\C_n$ define a double coset
\[
[U]^{\loc} := \C_n^0 U \C_n^0.
\]
It includes all elements of the Clifford group obtained from $U$ by adding single-qubit Clifford gates on the left and on the right.  Clearly, the full Clifford group $\C_n$ is a disjoint union of double cosets $[U]^{\loc}$ and the cost of the matrix $U$ depends only on the double coset that contains $U$.  The next step is to choose an efficiently computable canonical representative of each double coset.  

First define the map $\chi\, {:} \, \F_2^{2n} \to \F_2^n$ as 
\[
\chi(v):=[v_1 \lor v_{n+1}, \; v_2 \lor v_{n+2}, \ldots, v_{n} \lor  v_{2n}],
\]
where $\lor$ stands for the logical OR operation.
The $j$-th component of $\chi(v)$ is non-zero iff $v_j{=}1$ or $v_{n+j}{=}1$ (the bitstring $\chi(v)$ can be interpreted as the support of an $n$-qubit Pauli operator parameterized by $v$, according to the standard binary parameterization of Pauli operators \cite{aaronson2004improved}).  We claim that the map $\chi$ is invariant under left multiplications by the elements of the local subgroup, in the sense that
\be
\label{chi_invariant}
\chi(Lv) = \chi(v) \mbox{ for all $L\,{\in}\, \C_n^0$ and $v\,{\in}\, \F_2^{2n}$}.
\ee
Indeed, it suffices to check Eq.~(\ref{chi_invariant}) for the special case $L\,{\in}\, \C_{n,j}$ (since the local subgroup is generated by matrices $L\,{\in}\, \C_{n,j}$ with $j=1,2,\ldots,n$).  As discussed above, the action of $L\,{\in}\, \C_{n,j}$ on $v$ is equivalent to applying a $2{\times} 2$ binary invertible matrix to the components $v_j$ and $v_{n+j}$ while all other components of $v$ remain unchanged.  Since an invertible matrix maps nonzero vectors to nonzero vectors, $(Lv)_j \lor (Lv)_{n+j}\,{=}\,1$ iff $v_j \lor v_{n+j}\,{=}\,1$.  This implies Eq.~(\ref{chi_invariant}).

A matrix $U\,{\in}\, \C_n$ is said to be {\em locally ordered} if $U$ is left-reduced and
\begin{equation}
\label{locally_reduced}
\chi(U e^j) \le \chi(U e^{n+j}) \le  \chi(U e^j \oplus Ue^{n+j})  \mbox{ for all $j=1,2,\ldots,n$}.
\end{equation}
Here  bit strings are compared using the lexicographic order.  Let ${\cal L}(U)\,{\subseteq}\, [U]^{\loc}$ be the set of all locally ordered elements of the double coset $[U]^{\loc}$.  Define a {\em locally reduced} form of the matrix $U\,{\in}\, \C_n$, denoted $\mathsf{localReduce}(U)$, as the  lexicographically smallest element of the set ${\cal L}(U)$.  The following lemma shows that locally reduced elements of $\C_n$ can serve as canonical representatives of the double cosets $[U]^{\loc}$.  In other words, $\C_n$ is a disjoint union of the double cosets $[U]^{\loc}$ and each double coset contains a unique locally reduced element that can be efficiently computed (albeit slightly less efficiently than $\mathsf{leftReduce}$).  The symplectic matrix data structure described in \ssec{datastructure} enables the computation of $\mathsf{localReduce}(U)$ for a randomly picked matrix $U\,{\in}\, \C_n$ in time less than $4{\cdot}10^{-7}$ seconds for all $n\,{\le}\, 6$ on a server-class CPU, in this case an Intel\textsuperscript{\textregistered} Xeon\textsuperscript{\textregistered} CPU E7-4850 v4 @ 2.10GHz.

\begin{lemma}
\label{lem:locally_reduced}
Each double coset $[U]^{\loc}\,{=}\,\C_n^0 U \C_n^0$ contains a unique locally reduced element that can be computed in time $O(n^2 6^n)$, given the symplectic matrix $U$.
\end{lemma}

\begin{proof}
For each qubit $j$ define the bit strings $x_j\,{:=}\,\chi(Ue^j)$, $z_j\,{:=}\,\chi(U e^{n+j})$, and $y_j\,{:=}\,\chi(U e^j{\oplus}Ue^{n+j})$.  Same as before, multiplying $U$ on the right by the elements of the subgroup $\C_{n,j}$ one can implement any permutation of the bit strings $x_j$, $y_j$, and $z_j$. 
Define a subset $\S_j\,{\subseteq}\,\C_{n,j}$ as the one including all elements $R_j\,{\in}\, \C_{n,j}$ such that the right multiplication $U\,{\gets}\, UR_j$ permutes the bit strings $x_j$, $y_j$, and $z_j$ into the non-decreasing order $x_j\,{\le}\,z_j\,{\le}\,y_j$.  Note that $\S_j$ is non-empty since the right multiplication by the elements of $\C_{n,j}$ can implement any permutation of $x_j$, $y_j$, and $z_j$.  Recall that the set ${\cal L}(U)$ includes all locally ordered elements of the double coset $[U]^{\loc}$.  We claim that
\be
\label{L(U)}
{\cal L}(U)=\{ \mathsf{leftReduce}(UR_1R_2\cdots R_n)\, {:} \, R_1\,{\in}\,\S_1,  R_2\,{\in}\,\S_2, \ldots, R_n\,{\in}\,\S_n\}.
\ee
Indeed, ${\cal L}(U) \subseteq [U]^{\loc}$ since any matrix $W\,{\in}\,{\cal L}(U)$ has the form $W\,{=}\,LUR$ for some $L,R\,{\in}\, \C_n^0$.  Furthermore,  ${\cal L}(U)$ is non-empty since each subset $\S_j$ is non-empty.  Let us check that any element $W\,{\in}\, {\cal L}(U)$ is locally ordered.  Indeed, pick any matrices $R_j\,{\in}\, \S_j$ and let $R=R_1 R_2\cdots R_n$.  By construction, the matrix $V\,{=}\,UR$ satisfies Eq.~(\ref{locally_reduced}) with $U$ replaced by $V$.  Let $W\,{=}\,\mathsf{leftReduce}(V)$. Then $W\,{=}\,LV$ for some $L\,{\in}\,\C_n^0$.  The invariance of the map $\chi$ under left multiplications by the elements of the local subgroup, see Eq.~(\ref{chi_invariant}),  implies that $W$ satisfies Eq.~(\ref{locally_reduced}) with $U$ replaced by $W$.  Thus $W$ is locally ordered.  Conversely, suppose $W\,{\in}\, [U]^{\loc}$ is locally ordered.  Then $W\,{=}\,LUR$ for some $L,R\,{\in}\,\C_n^0$ and $\mathsf{leftReduce}(W)\,{=}\,W$.  The invariance of the map $\chi$ under left multiplications by the elements of the local subgroup and the local ordering condition imply that the matrix $V\,{=}\,UR$ satisfies Eq.~(\ref{locally_reduced}) with $U$ replaced by $V$.  Thus $R\,{=}\,R_1R_2\cdots R_n$ for some $R_j\,{\in}\,\S_j$. This proves that $W\,{\in}\,{\cal L}(U)$. The uniqueness follows from the ability to encode the elements of the sets considered by distinct integers and the existence of the smallest integer in any finite set of integers.

It remains to check that the set ${\cal L}(U)$ can be computed in time $O(n^2 6^n)$.  Indeed, for any given qubit $j$ one can compute the bit strings $x_j$, $y_j$, and $z_j$ and the subset $S_j \,{\subseteq}\, \C_{n,j}$ in time $O(n)$.  Note that $|\S_j|\,{\le}\, |\C_{n,j}|\,{=}\,6$.  Thus the number of matrices $R\,{=}\,R_1R_2\cdots R_n$ with $R_j\,{\in}\, \S_j$ is at most $6^n$.  Since the right multiplication by the elements of the subgroup $\C_{n,j}$ changes at most two rows of a matrix, we can compute $UR$ in time $O(n^2)$.  By \lem{left_reduced}, computing the left reduced form of $UR$ takes time $O(n^2)$.  Thus the overall runtime of computing ${\cal L}(U)$ is $O(n^2 6^n)$.  Once the set ${\cal L}(U)$ is computed, finding its lexicographically smallest element takes time $O(n|{\cal L}(U)|)\,{=}\, O(n6^n)$.
\end{proof}

{\em Comment 1:} 
Our implementation of $\mathsf{localReduce}(U)$ relies on a streamlined version of the  above algorithm with a modified definition of the subsets $\S_j$. Namely, we define $\S_j$ as a set of all elements $R_j\,{\in}\, \C_{n,j}$ such that the right multiplication $U\,{\gets}\, UR_j$ permutes the bit strings $x_j$, $y_j$, and $z_j$ into the non-decreasing order and $\mathsf{leftReduce}(UR_j)\,{\ne}\, \mathsf{leftReduce}(U)$.  The last condition rules out the possibility that the right multiplication of $U$ by $R_j$ is equivalent to a left multiplication of $U$ by some element of the local subgroup (for example, this is the case if $U$ is the identity matrix).  Since $\mathsf{leftReduce}(U)$ depends only on the coset $\C_n^0 U$, the left multiplication of $U$ by any element of the local subgroup does not change  $\mathsf{leftReduce}(U)$.  Thus the set of locally ordered elements ${\cal L}(U)$ can be computed
using Eq.~(\ref{L(U)}) with the modified definition of $\S_j$.

{\em Comment 2:} 
We empirically observed that the average-case runtime of the above algorithm is much better than the worst case upper bound of $O(n^2 6^n)$.  Indeed, a direct inspection shows that the runtime scales as $O(n^2M)$, where $M\,{=}\,|\S_1|{\cdot}|\S_2|{\cdot}\ldots{\cdot}|\S_n|$.  For randomly picked matrices $U\,{\in}\, \C_6$ we observed that $M\,{\approx}\,5$ on average even though $M\,{=}\,|\C_6^0|\,{=}\,6^6\,{=}\,46{,}656$ in the worst case.  We leave it as an open question whether the average-case runtime of the above algorithm scales polynomially with $n$. 

Recall that we consider the symmetric group $S_n$ that includes all qubit permutations as a subgroup of $\C_n$.  If $w$ is a permutation of integers $\{1,2,\ldots,n\}$, then the corresponding symplectic matrix $W\,{\in}\, S_n$ acts on the basis vectors as $W e^j\,{=}\,e^{w(j)}$ and $W e^{n+j}\,{=}\,e^{n+w(j)}$ for all $j\,{=}\,1,2,\ldots,n$.  Given a matrix $U\,{\in}\, \C_n$, define the equivalence class
\[
[U] := \{ LW^{-1}UW R\,{:}\, L,R\in \C_n^0, W\,{\in}\,S_n\}.
\] 
The rest of this section is devoted to choosing an efficiently computable canonical representative of each class $[U]$.  Let $\ZZ^{n\times n}$ be the set of $n{\times}n$ matrices with integer entries.  Define the map $\kappa\,{:}\, \C_n \,{\to}\, \ZZ^{n\times n}$ such that the matrix element of $\kappa(U)$ located at the $i$-th row and the $j$-th column is the rank of the $2{\times}2$ submatrix of $U$ formed by the intersection of rows $i$ and $i{+}n$ and columns $j$ and $j{+}n$.  The rank is computed over the binary field $\F_2$.  In other words, each matrix element of $\kappa(U)$ has the form
\[
\kappa(U)_{i,j} = \mathrm{rank}_{\F_2}\left[ \begin{array}{ll}
U_{i,j} & U_{i,n+j} \\
U_{n+i,j} & U_{n+i,n+j} \\
\end{array}
\right].
\]
By definition, $\kappa(U)$ contains entries from the set $\{0,1,2\}$ and the full matrix $\kappa(U)$ can be computed in time $O(n^2)$.  We claim that the left and right multiplications of $U$ by the single-qubit Clifford gates leave $\kappa(U)$ invariant, that is, 
\be
\label{kappa_invariant}
\kappa(LUR)=\kappa(U) \mbox{ for all  $L,R\,{\in}\, \C_n^0$}.
\ee
Indeed, suppose first that $L{=}I$ and $R\,{\in}\,\C_{n,j}$. Right multiplication $U\,{\gets}\, UR$ applies an invertible linear transformation to the pair of columns $Ue^j$ and $Ue^{n+j}$, and acts trivially on the remaining columns.  Since the matrix rank is invariant under applying an invertible linear transformation, we conclude that $\kappa(UR)\,{=}\,\kappa(U)$ for all $R\,{\in}\,\C_{n,j}$.  Same argument shows that $\kappa(LU)\,{=}\,\kappa(U)$ for all $L\,{\in}\, \C_{n,j}$.  This proves Eq.~(\ref{kappa_invariant}) since the local subgroup $\C_n^0$ is generated by the subgroups $\C_{n,j}$.

Let $\kappa_{min}(U)$ be the lexicographically smallest matrix in the set of matrices $\{ \kappa(W^{-1} U W)\,{:}\, W\,{\in}\,S_n\}$.  Define a set of qubit permutations 
\[
{\cal S}(U):= \{ W \in S_n \,{:}\, \kappa(W^{-1} U W)=\kappa_{min}(U)\}
\]
and a set of matrices
\[
{\cal R}(U):= \{ \mathsf{localReduce}(W^{-1} U W) \,{:}\, W \in {\cal S}(U) \}.
\]
Note that ${\cal R}(U)\,{\subseteq}\,[U]$ since 
\[
\mathsf{localReduce}(W^{-1} U W) = L W^{-1} U W R \in [U]
\]
for some $L,R\in \C_n^0$.  Define a {\em fully reduced} form of a matrix $U\,{\in}\, \C_n$, denoted $\mathsf{ReduceU}(U)$, as the lexicographically smallest element of the set ${\cal R}(U)$.  The following lemma shows that the fully reduced elements of $\C_n$ can serve as canonical representatives of the equivalence classes $[U]$.  In other words, $\C_n$ is a disjoint union of the equivalence classes $[U]$ and each class contains a unique fully reduced element that can be efficiently computed (albeit slightly less efficiently than $\mathsf{localReduce}$).  The symplectic matrix data structure enables the computation of $\mathsf{ReduceU}(U)$ for a randomly picked matrix $U\,{\in}\,\C_n$ in time less than $3{\cdot}10^{-6}$ seconds for $n{=}6$ and time less than $10^{-6}$ seconds for all $n\,{\le}\, 5$ on a server-class CPU, in this case an Intel\textsuperscript{\textregistered} Xeon\textsuperscript{\textregistered} CPU E7-4850 v4 @ 2.10GHz.

\begin{lemma}
Each equivalence class $[U]$ with $U\,{\in}\,\C_n$ contains a unique fully reduced element that can be computed in time $O(n^2{\cdot}n!  + t_n{\cdot}|{\cal S}(U)|)$, given the symplectic matrix representation of $U$.  Here $t_n$ is the runtime of $\mathsf{localReduce}$ for elements of $\C_n$.
\end{lemma}

\begin{proof}
Consider a matrix $U\,{\in}\, \C_n$.  It follows directly from the definitions that $\mathsf{ReduceU}(U)\,{\in}\,[U]$.  Thus it suffices to check that 
\be
\label{claim1}
{\cal R}(U')= {\cal R}(U) \mbox{ for all $U'\in [U]$}.
\ee
Indeed, this equation implies $\mathsf{ReduceU}(U)\,{=}\,\mathsf{ReduceU}(U')$ for all  $U'\,{\in}\,[U]$, that is, the equivalence class $[U]$ contains a unique reduced element.  Let us prove Eq.~(\ref{claim1}).  Write $U'\,{=}\,L W^{-1} U W R$ for some $L,R\,{\in}\, \C_n^0$ and $W\,{\in}\,S_n$. Then
\begin{align}
{\cal R}(U')&=\{ \mathsf{localReduce}(\tilde{W}^{-1} L W^{-1} U W  R \tilde{W}) \,{:}\, \tilde{W} \in {\cal S}(U') \} \nonumber \\
& = \{ \mathsf{localReduce}(L' \tilde{W}^{-1} W^{-1} U W  \tilde{W} R') \,{:}\, \tilde{W} \in {\cal S}(U') \} \nonumber \\
& = \{ \mathsf{localReduce}( \tilde{W}^{-1} W^{-1} U W  \tilde{W} ) \,{:}\, \tilde{W} \in {\cal S}(U') \}. \label{RU1}
\end{align}
Here $L'\,{:=}\,\tilde{W}^{-1} L \tilde{W}\,{\in}\, \C_n^0$ and $R'\,{:=}\,\tilde{W}^{-1} R \tilde{W}\,{\in}\,\C_n^0$.  In the third equality we noted that $\mathsf{localReduce}$ is invariant under left/right multiplications by the elements of the local subgroup $\C_n^0$, see  \lem{locally_reduced}.  Finally, the invariance of the map $\kappa$ under the left and right multiplications by the elements of the local subgroup, see Eq.~(\ref{kappa_invariant}), implies $\kappa_{min}(U')=\kappa_{min}(U)$.  Thus $\tilde{W}\,{\in}\,{\cal S}(U')$ iff $W\tilde{W} \,{\in}\,{\cal S}(U)$.  Combining this and Eq.~(\ref{RU1}) gives ${\cal R}(U')\,{=}\,{\cal R}(U)$, as claimed. 

The runtime stated in the lemma consists of two terms. The term $O(n^2{\cdot}n!)$ is the time needed to compute the set of permutations ${\cal S}(U)$. The term $O(t_n {\cdot}|{\cal S}(U)|)$ is the time needed to compute the set of matrices ${\cal R}(U)$ and pick the lexicographically smallest element of ${\cal R}(U)$.
\end{proof}

{\em Comment 3:} 
Our implementation of $\mathsf{ReduceU}(U)$ relies on a streamlined version of the above algorithm with a modified definition of the set ${\cal S}(U)$. Namely, we define ${\cal S}(U)$ as the set of all permutations $W\,{\in}\, S_n$ such that $\kappa(W^{-1} U W)\,{=}\,\kappa_{min}(U)$ and $\mathsf{leftReduce}(W^{-1} U W)\,{\ne}\,\mathsf{leftReduce}(U)$.  The last condition rules out the possibility that the conjugation of $U$ by $W$ is equivalent to a left multiplication of $U$ by some element of the local subgroup (for example, this is the case if $U$ is the identity matrix).  Since $\mathsf{localReduce}(U)$ depends only on the double coset $\C_n^0 U\C_n^0$, a left multiplication of $U$ by any element of the local subgroup does not change $\mathsf{localReduce}(U)$.  Thus one can compute the set ${\cal R}(U)$ using the modified definition of ${\cal S}(U)$.

{\em Comment 4:}
We empirically observed that $|{\cal S}(U)|{=}1$ for typical a element of the Clifford group and the maximal value of $|{\cal S}(U)|$ is $14$.  The mean value of $|{\cal S}(U)|$ is approximately $1.03$ for a randomly picked $U\,{\in}\,\C_6$. 
 
By a slight abuse of terminology, we refer to the computationally-defined fully reduced elements of the Clifford group as the reduced elements in the remainder of the paper.  This should not lead to confusion since the left-reduced and the locally reduced forms are used only in this subsection.

\subsection{Data structure}\label{ssec:datastructure}

By definition, any element of the Clifford group  $U\,{\in}\,\C_n$ can be represented by a binary matrix of size $2n{\times}2n$.  However, if we only care about the reduced form of $U$, a slightly more efficient representation is possible, as given by the following lemma. 

\begin{lemma}
Let $U'$ be the matrix obtained from $U\,{\in}\,\C_n$ by removing the $n$-th and the $2n$-th rows from it.  Then $U$ is uniquely determined by $U'$ up to left multiplication by the single-qubit Clifford gates acting on the $n$-th qubit. 
\end{lemma}

\begin{proof}
Let ${\cal L}\,{\subseteq}\,\F_2^{2n}$ be the linear subspace spanned by the $j$-th row of $U$ with $j\,{\notin}\,\{n,2n\}$ and let ${\cal L}^\perp\,{\subseteq}\,\F_2^{2n}$ be the linear subspace spanned by the vectors orthogonal to ${\cal L}$ with respect to the symplectic inner product. Note that ${\cal L}$ depends only on $U'$.  The condition that $U$ is a symplectic matrix implies $\mathrm{span}_{\F_2}( e_n U, e_{2n} U) \,{=}\, {\cal L}^\perp$.  Here we use the notations from \ssec{reduceu}.  The missing pair of rows $e_nU$ and $e_{2n}U$ is uniquely defined by ${\cal L}$ up to an invertible linear transformation $e_nU\gets a e_n U \,{\oplus}\, b e_{2n}U$ and $e_{2n}U\gets c e_n U \,{\oplus}\, d e_{2n}U$ for some
\[
\left[ \ba{cc}
a & c \\
b & d \\
\ea \right] \in  \mathrm{GL}(2,\F_2).
\]
As discussed in \ssec{reduceu}, there is a one-to-one correspondence between such transformations and left multiplications $U\,{\gets}\,LU$, where $L\,{\in}\,\C_{n}^0$ acts non-trivially only on the $n$-th qubit.
\end{proof}

We refer to the matrix $U'$ obtained from $U\,{\in}\,\C_n$ by removing the pair of rows $n$ and $2n$ as a {\em thin matrix} representation of $U$.  Our C\texttt{++} implementation adopts the thin matrix data format for all intermediate steps of the algorithm.  The thin matrix spans $4n(n{-}1)$ bits and can be conveniently distributed over two machine words, each of length $64$ bits.  The first word stores the rows $e_1U$, $e_2U$, $\ldots, e_{n-1}U$ and the second word stores the rows $e_{n+1}U$, $e_{n+2}U$, $\ldots, e_{2n-1}U$.  This leaves $128\,{-}\,4n(n{-}1)|_{n{\leq}6} \geq 8$ free bits that can be conveniently used to specify the cost-reducing generator in the augmented database, see \ssec{database}.  Recall that the number of generators is $m\,{=}\,9n(n{-}1)/2|_{n{\leq}6}\le 135$. Thus the generator can be specified using only $8$ bits.  Note also that storing the full matrix $U\,{\in}\,\C_n$ using only two machine words is impossible for $n{=}6$, as it requires $4n^2|_{n{=}6} \,{=}\, 144$ bits.

The thin matrix format enables fast left and right multiplication by the single-qubit and two-qubit Clifford gates, that require at most 24 CPU instructions per gate for all $n{\le} 6$  (each instruction implements a bitwise operation on a single machine word).  When needed, the thin matrix $U'$ can be expanded into the full symplectic matrix $U\,{\in}\,\C_n$ by calculating the missing pair of rows $e_nU$ and $e_{2n}U$ using the symplectic version of Gram-Schmidt orthogonalization.  Our implementation converts the thin matrix to the full matrix in time less than $2{\cdot}10^{-7}$ seconds for any $n{\le}6$ on a server-class CPU, in this case an Intel\textsuperscript{\textregistered} Xeon\textsuperscript{\textregistered} CPU E7-4850 v4 @ 2.10GHz, which is negligible compared with the time it takes to compute the reduced form.

\subsection{Software tricks}\label{ssec:sw}

\textbf{Database generation:} 
The calculation of the reduced cost-$k$ Clifford group set $\R_n^k$, as described in \ssec{database}, lends itself to parallel processing.  Specifically, each element of the set $\R_n^k$ can be calculated concurrently from its own data on its own processor.  The implementation considerations for this run-once parallel processing job depended on factors such as: 
\renewcommand{\theenumi}{\roman{enumi}}%
\begin{enumerate}
  \item the cost and availability of scaled-up/scaled-out hardware, and
  \item the cost-benefit for implementing, measuring, and tuning for different data-level parallel processing options, including shared memory versus distributed memory (e.g., OpenMP/MPI) and specialized processors (e.g., vector processors, GPUs, FPGAs),
\end{enumerate}
not to mention the multiple software options with each, from programming languages to libraries\cite{clang}.

Using Flynn's taxonomy \cite{flynn}, the \textit{Single Program, Multiple Data} (SPMD) \textit{streams} model was implemented using the C\texttt{++} \textit{concurrent-set} template class; specifically, each reduced cost-$k$ Clifford group set $\R_n^k$ is an instance of  \textascii{set<pair<uint64, uint64>>}.
This is a good choice for programmer productivity, i.e., letting the container's semantics deal with the requirements of maintaining distinct and efficiently-searchable elements of a multi-terabyte set on SMP hardware, in this case an Intel\textsuperscript{\textregistered} Xeon\textsuperscript{\textregistered} 128-CPU E7-4850 v4 @ 2.10GHz with 6TB RAM.

Runtime was extrapolated to take about $100$ days to complete the full database generation on a single machine, amounting to approximately $100{\cdot}24{\cdot}128 \,{=}\, 307{,}200$ CPU-hours that can be effectively divided among as many machines as there are available.  Hardware and software measurements during database generation, using performance analysis tools such as \textit{vmstat} to VTune\texttrademark, exposed heavy ``NUMA thrashing,'' i.e., soft page faults \cite{numathrashing}.  To alleviate this for the final half of the run, C's most basic systems programming mechanisms were more readily and easily used to replace the C\texttt{++} \textit{set} template in order to allocate, position, and search raw memory, resulting in a 5x speed-up; namely, \textit{malloc}, \textit{bsearch}, and \textit{qsort}, along with \textit{read/write} and \textit{uint128}.

\begin{table}[t]
\centering
\begin{minipage}[c]{0.85\linewidth}
\centering
\begin{tabular}{|r|r|r|r|r|r|} \hline
\backslashbox{$\cnotgate$ count}{Qubits} & 2 & 3 & 4 & 5 & 6  \\ \hline
0 & 1 & 1 & 1 & 1 & 1 \\
1 & 1 & 1 & 1 & 1 & 1 \\
2 & 1 & 3 & 4 & 4 & 4 \\ 
3 & 1 & 8 & 20 & 22 & 23 \\ 
4 && 10 & 112 & 183 & 198 \\ 
5 && 3 & 525 & 1,958 & 2,549 \\
6 && 1 & 1,230 & 22,257 & 42,883 \\
7 &&& 453 & 223,723 & 824,723 \\
8 &&& 16 & 1,441,124 & 16,086,167 \\
9 &&& 1 & 2,471,855 & 294,266,642 \\
10 &&&& 161,458 & 4,399,997,085 \\
11 &&&& 72 & 40,791,942,327 \\
12 &&&& 1 & 92,804,759,960 \\
13 &&&&& 5,666,221,415 \\
14 &&&&& 8,281 \\
15 &&&&& 3 \\ \hline
Total & 4 &	27 & 2,363 & 4,322,659 & 143,974,152,262 \\ \hline
\end{tabular}
\caption{The distribution of the number of equivalence classes across Clifford circuits over $2$, $3$, $4$, $5$, and $6$ qubits.}
\label{tab:2345}
\end{minipage}
\end{table}

\textbf{Synthesis of optimal circuits:}
With the one-time generation of the database complete and saved on secondary storage (Solid State Disk), similar systems programming mechanisms in C were exploited to optimize performance and scalability in order to read/search what is now effectively a lookup table (LUT), with the expensive runtime calculation of an optimal 6-qubit Clifford circuit completed and replaceable by a simple array indexing operation.  The database can be memory-mapped with \textit{mmap} \cite{mmap} for a greater degree of
\renewcommand{\theenumi}{\roman{enumi}}%
\begin{enumerate}
  \item programmer productivity, i.e., the database can be easily  referenced as memory using pointers, with no explicit file IO, and
  \item operational flexibility, i.e., the database can be effectively used by any type of hardware, ranging from a single laptop to a cluster of server-class machines, with scaling solely dependent on the choice of hardware, 
\end{enumerate}
all without changing the code; while the OS kernel and \textit{mmap} transparently and efficiently take care of
\renewcommand{\theenumi}{\roman{enumi}}%
\begin{enumerate}
  \item demand paging, and
  \item maintaining only a single copy of data in memory, as opposed to copies in both the file cache and user space. 
\end{enumerate}
In addition, to reduce the number of SSD queries, being the most time-consuming operation our search relies on, we employed the following strategy: 
\renewcommand{\theenumi}{\roman{enumi}}%
\begin{enumerate}
  \item we store the databases of Clifford circuits requiring 1--8, 14, and 15 gates in RAM, 
  \item we store an index consisting of each $1024^\text{th}$ element of Clifford unitaries implementable with 9--13 gates in RAM, and
  \item when the length-1024 chunk containing the desired element is found by the binary search, we make one long query to extract all 2048 64-bit integers in this chunk.
\end{enumerate}
The above modification limits the number of SSD queries required to synthesize an optimal circuit to at most 10 (at most two queries per searches over the gate counts of 9, 10, 11, 12, and 13) at the cost of RAM memory usage of $2.5$GB.

A machine with enough RAM to fit the entire database in will get the best performance as the complete database fills the file cache, and a machine with little-to-no available RAM will get the worst performance as every pointer access to a memory-mapped region (e.g., \textit{bsearch}) will touch the secondary storage.  A commodity machine with typical RAM sizes will get near-best performance as the ``hot'' parts of the database\textemdash{}the internal nodes of \textit{bsearch}\textemdash{}will tend to remain in the cache hierarchy (L1-L3, file cache) and result in minimal access to secondary storage. OS-specific parameters were not explored but can also be benchmarked and tuned independently of the database and code, including page sizes and pinned memory.

\section{Results}\label{sec:results} 

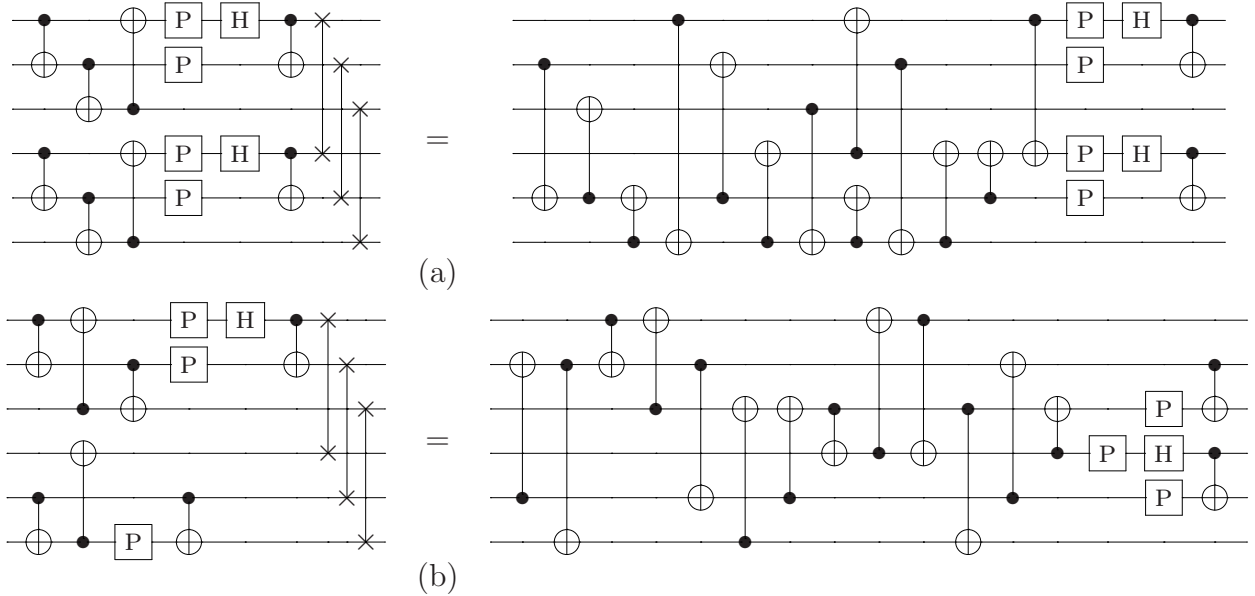
\begin{figure}[t]
\centering
\begin{minipage}[c]{0.85\linewidth}
\begin{center}
\begin{tabular}{ccc}
\Qcircuit @C=0.6em @R=0.3em @!R {
&\ctrl{1}&\qw       &\targ     &\gate{\pgate} &\gate{\hgate} &\ctrl{1} &\qswap    &\qw       &\qw        & \qw \\
&\targ   &\ctrl{1}  &\qw       &\gate{\pgate} &\qw           &\targ    &\qw\qwx   &\qswap    &\qw        & \qw \\
&\qw     &\targ     &\ctrl{-2} &\qw           &\qw           &\qw      &\qw\qwx   &\qw\qwx   &\qswap     & \qw \\
&\ctrl{1}&\qw       &\targ     &\gate{\pgate} &\gate{\hgate} &\ctrl{1} &\qswap\qwx&\qw\qwx   &\qw\qwx    & \qw \\
&\targ   &\ctrl{1}  &\qw       &\gate{\pgate} &\qw           &\targ    &\qw       &\qswap\qwx&\qw\qwx    & \qw \\
&\qw     &\targ     &\ctrl{-2} &\qw           &\qw           &\qw      &\qw       &\qw       &\qswap\qwx & \qw 
}
& \raisebox{-17mm}{=} &
\Qcircuit @C=0.6em @R=0.3em @!R {
&\qw     &\qw      &\qw      &\ctrl{5} &\qw      &\qw      &\qw     &\targ    &\qw     &\qw      &\qw      &\ctrl{3}&\gate{\pgate}&\gate{\hgate}&\ctrl{1}&\qw \\
&\ctrl{3}&\qw      &\qw      &\qw      &\targ    &\qw      &\qw     &\qw      &\ctrl{4}&\qw      &\qw      &\qw     &\gate{\pgate}&\qw          &\targ   &\qw \\
&\qw     &\targ    &\qw      &\qw      &\qw      &\qw      &\ctrl{3}&\qw      &\qw     &\qw      &\qw      &\qw     &\qw          &\qw          &\qw     &\qw \\
&\qw     &\qw      &\qw      &\qw      &\qw      &\targ    &\qw     &\ctrl{-3}&\qw     &\targ    &\targ    &\targ   &\gate{\pgate}&\gate{\hgate}&\ctrl{1}&\qw \\
&\targ   &\ctrl{-2}&\targ    &\qw      &\ctrl{-3}&\qw      &\qw     &\targ    &\qw     &\qw      &\ctrl{-1}&\qw     &\gate{\pgate}&\qw          &\targ   &\qw \\
&\qw     &\qw      &\ctrl{-1}&\targ    &\qw      &\ctrl{-2}&\targ   &\ctrl{-1}&\targ   &\ctrl{-2}&\qw      &\qw     &\qw          &\qw          &\qw     &\qw 
}
\\
& (a) & \\
\Qcircuit @C=0.6em @R=0.3em @!R {
&\ctrl{1}  &\targ       &\qw         &\gate{\pgate} &\gate{\hgate} &\ctrl{1}  &\qswap    &\qw       &\qw        & \qw \\
&\targ     &\qw  &\ctrl{1}          &\gate{\pgate} &\qw          &\targ     &\qw\qwx   &\qswap    &\qw        & \qw \\
&\qw       &\ctrl{-2}     &\targ     &\qw           &\qw           &\qw       &\qw\qwx   &\qw\qwx   &\qswap     & \qw \\
&\qw       &\targ  &\qw  &\qw &\qw     &\qw    &\qswap\qwx&\qw\qwx   &\qw\qwx    & \qw \\
&\ctrl{1}  &\qw       &\qw           &\ctrl{1}           &\qw           &\qw &\qw       &\qswap\qwx&\qw\qwx    & \qw \\
&\targ     &\ctrl{-2}     &\gate{\pgate}            &\targ &\qw        &\qw       &\qw       &\qw       &\qswap\qwx & \qw 
}
& \raisebox{-17mm}{=} &
\Qcircuit @C=0.6em @R=0.3em @!R {
& \qw & \qw & \ctrl{1} & \targ & \qw & \qw & \qw & \qw & \targ & \ctrl{3} & \qw & \qw & \qw & \qw & \qw & \qw & \qw \\
& \targ & \ctrl{4} & \targ & \qw & \ctrl{3} & \qw & \qw & \qw & \qw & \qw & \qw & \targ  & \qw & \qw & \qw & \ctrl{1} & \qw \\
& \qw & \qw & \qw & \ctrl{-2} & \qw & \targ & \targ & \ctrl{1} & \qw & \qw & \ctrl{3} & \qw & \targ & \qw & \gate{\pgate} & \targ & \qw \\
& \qw & \qw & \qw & \qw & \qw & \qw & \qw & \targ & \ctrl{-3} & \targ & \qw & \qw & \ctrl{-1} & \gate{\pgate} &\gate{\hgate} & \ctrl{1}  & \qw \\
& \ctrl{-3} & \qw & \qw & \qw & \targ & \qw & \ctrl{-2} & \qw & \qw & \qw & \qw & \ctrl{-3} & \qw & \qw & \gate{\pgate} & \targ & \qw \\
& \qw & \targ & \qw & \qw & \qw & \ctrl{-3} & \qw & \qw & \qw & \qw &\targ & \qw & \qw & \qw & \qw & \qw & \qw \\
} \\
& (b) & \\
\end{tabular}
\end{center}
\caption{All most expensive 6-qubit Clifford unitaries requiring $15$ entangling gates (up to left and right multiplication by the single-qubit gates and qubit relabeling).  (a) left: a compact representation in the form $(U \otimes U)\text{SWAP}$, right: its optimal implementation; (b) left: a compact representation in the form $(U' \otimes V')\text{SWAP}$, right: its optimal implementation.  Not illustrated is the cyclic SWAP of all $6$ qubits, that also requires $15$ entangling gates.} 
\label{fig:6q15g}
\end{minipage}
\end{figure}

The distribution of the number of equivalence classes across $\cnotgate$ gate costs is shown \tab{2345}.  For the number of qubits $2$ through $5$ the most complex function to implement is unique (within the equivalence class definition), and it is equivalent to a cyclic permutation of qubits.  For $n{=}6$, the cyclic permutation is one of three such functions; the other two are illustrated in \fig{6q15g}.  The small number of equivalence classes for a small number of qubits implies an efficient formula (based on $\mathsf{ReduceU}$) to compute the $\cnotgate$ cost of a small Clifford unitary.

\begin{table}[t]
\centering
\begin{minipage}[c]{0.85\linewidth}
\centering
\begin{tabular}{|r|r|} \hline
$\cnotgate$ cost & Number of $6$-qubit Clifford unitaries  \\ \hline
0 & 46,656 \\
1 & 6,298,560 \\
2 & 554,273,280 \\
3 & 39,045,473,280 \\
4 & 2,365,081,986,240 \\
5 & 126,526,140,927,360 \\
6 & 5,998,793,185,860,480 \\
7 & 249,378,588,704,827,008 \\
8 & 8,870,235,256,471,637,952 \\
9 & 255,646,483,904,239,690,752 \\
10 & 5,278,109,585,506,533,785,088 \\
11 & 58,697,087,161,047,579,538,560 \\
12 & 135,876,260,385,953,644,020,480 \\
13 & 7,998,401,853,543,422,302,848 \\
14 & 6,525,042,824,342,016 \\
15 & 13,308,157,440  \\ \hline
 & 208,114,637,736,580,743,168,000 \\ \hline
\end{tabular}
\caption{The distribution of the number of $6$-qubit Clifford unitaries across the entangling gate cost.}
\label{tab:6qdistribution}
\end{minipage}
\end{table}

We ran a script to calculate the distribution of the number of Clifford group elements across optimal $\cnotgate$ gate costs.  Given the database, it took a few days to collect the data using an HPC system.  This computation is highly parallelizable, and the runtime can be reduced significantly with many processors, e.g., GPUs; we have not pursued those reductions.  The results are reported in \tab{6qdistribution}.

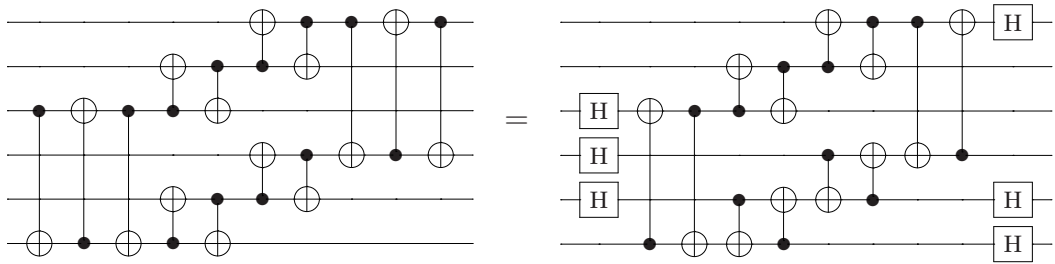
\begin{figure}[t]
\centering
\begin{minipage}[c]{0.85\linewidth}
\begin{center}
\begin{tabular}{ccc}
\Qcircuit @C=0.6em @R=0.6em @!R {
& \qw       & \qw       & \qw       & \qw       & \qw       & \targ       & \ctrl{1}  & \ctrl{3}  & \targ     & \ctrl{3} & \qw\\
& \qw       & \qw       & \qw       & \targ     & \ctrl{1}  & \ctrl{-1}   & \targ     & \qw       & \qw       & \qw & \qw\\
& \ctrl{3}  & \targ     & \ctrl{3}  & \ctrl{-1} & \targ     & \qw         & \qw       & \qw       & \qw       & \qw & \qw\\
& \qw       & \qw       & \qw       & \qw       & \qw       & \targ       & \ctrl{1}  & \targ     & \ctrl{-3} & \targ & \qw\\
& \qw       & \qw       & \qw       & \targ     & \ctrl{1}  & \ctrl{-1}   & \targ     & \qw       & \qw       & \qw & \qw\\
& \targ     & \ctrl{-3} & \targ     & \ctrl{-1} & \targ     & \qw         & \qw       & \qw       & \qw       & \qw  & \qw 
}
&
\raisebox{-14mm}{=}
&
\Qcircuit @C=0.6em @R=0.3em @!R {
& \qw       & \qw       & \qw       & \qw       & \qw       & \targ       & \ctrl{1}  & \ctrl{3}  & \targ     & \gate{\hgate} & \qw\\
& \qw       & \qw       & \qw       & \targ     & \ctrl{1}  & \ctrl{-1}   & \targ     & \qw       & \qw       & \qw & \qw\\
& \gate{\hgate}  & \targ     & \ctrl{3}  & \ctrl{-1} & \targ     & \qw         & \qw       & \qw       & \qw       & \qw & \qw\\
& \gate{\hgate}       & \qw       & \qw       & \qw       & \qw       & \ctrl{1}       & \targ  & \targ     & \ctrl{-3} & \qw & \qw\\
& \gate{\hgate}       & \qw       & \qw       & \ctrl{1}     & \targ  & \targ   & \ctrl{-1}     & \qw       & \qw       & \gate{\hgate} & \qw\\
& \qw     & \ctrl{-3} & \targ     & \targ & \ctrl{-1}    & \qw         & \qw       & \qw       & \qw       & \gate{\hgate}  & \qw
}
\end{tabular}
\end{center}
\caption{An optimal $\cnotgate$ gate circuit (left) can be implemented with fewer entangling gates as an optimal Clifford circuit (right).}
\label{fig:cnotadvantage}
\end{minipage}
\end{figure}

We used the database to look for examples of quantum Clifford advantage over classical reversible $\cnotgate$ circuits, meaning optimal $\cnotgate$ circuits that can be implemented with fewer entangling gates as a Clifford circuit.  We found one such example, illustrated in \fig{cnotadvantage}, that gives a reduction of $14$ gates into $12$, improving the $8$ to $7$ reduction seen earlier \cite{bravyi2020hadamard} $\left(\text{indeed, } \frac{14}{12}{>}\frac{8}{7}\right)$. 

The compiler was benchmarked using both consumer-grade and enterprise-grade systems for a test set with $10{,}000$ elements of the Clifford group $\C_6$.  Each element was generated by a Clifford circuit with $600$ randomly chosen gates over the library $\{\hgate, \pgate,  \cnotgate\}$.  The number of gates was selected to be high enough to effect a close to random uniform distribution over the elements of the group $\C_6$.  We observed that such random test set is dominated by the elements with costs $11$ and $12$.  The compiler runtime reported below is the time required to obtain optimal circuits for all test set elements divided by the size of the test set.  We observed the runtime of $\laptophottime$ seconds for a laptop with Intel\textsuperscript{\textregistered} i7-1068NG7 2.3GHz CPU and 16GB RAM with USB-C-attached consumer-grade SSD.  The search relies on the database stored on SSD, and a $2.5$GB index in RAM, see \ssec{sw} for details. The time reported measures hot cache performance, cold cache performance reads $0.003708$ seconds per an optimal circuit, on average.  The compiler performance improves when the entire database can be stored in RAM.  We observed the hot cache runtime of approximately $\serverhottime$ seconds for a server with Intel\textsuperscript{\textregistered} Xeon\textsuperscript{\textregistered} 128-CPU E7-4850 v4 @ 2.10GHz and 6TB RAM.  The process of loading the full database into RAM took approximately 2 hours. 

This performance allows to use our implementation to obtain individual circuits and entire randomized benchmarking schedules in mere seconds using consumer-grade hardware as well as online via a web interface.  For the use in demanding applications such as peep-hole optimization of large circuits, we suggest relying on large-RAM commercial-grade servers and note that it takes roughly half the time to look up the cost without computing the optimal circuit (the procedure that would likely get called most frequently during peep-holing).

The average runtime of our compiler for random $n$-qubit Clifford operators with $n\,{\le}\,5$ is shown in \tab{runtime_small}. 

\begin{table}[t]
\centering
\begin{minipage}[c]{0.85\linewidth}
\centering
\begin{tabular}{|r|r|r|} 
\hline
Qubits & Average runtime (seconds) & Database size (bytes) \\
\hline
$n\,{=}\,5$ & 0.0002922 & 69,162,544 \\
$n\,{=}\,4$ & 0.0001928 & 37,808 \\
$n\,{=}\,3$ & 0.0001351 & 432 \\
$n\,{=}\,2$ & 0.00007968 & 64 \\
\hline
\end{tabular}
\caption{Average runtime for optimally compiling $n$-qubit Clifford operators with the full database of reduced elements loaded into RAM. The runtime was measured on MacBook Pro laptop (early 2015 model) with Intel\textsuperscript{\textregistered} i7-5557U 3.1GHz CPU and 16GB RAM.}
\label{tab:runtime_small}
\end{minipage}
\end{table}


\subsection{Optimal 2-designs}
\label{sec:2design}

Unitary designs~\cite{low2010pseudo} are probability distributions on the unitary group that reproduce low-order moments of the Haar (uniform) distribution.  Of particular interest are unitary designs that can be efficiently implemented by quantum circuits~\cite{cleve2015near}.  Such designs can serve as a substitute for the Haar distribution in certain randomized quantum protocols such as data hiding~\cite{divincenzo2002quantum}, estimating fidelity of quantum operations~\cite{magesan2011scalable, emerson2005scalable}, and quantum state tomography~\cite{huang2020predicting}.  In this section, we leverage the database of reduced Clifford elements to  construct optimal unitary designs that have the minimum average cost, subject to the constraint that all elements of the design are Clifford operators.

Let $U(2^n)$ be the group of unitary complex matrices of size $2^n{\times} 2^n$. Suppose  $\D \,{\subseteq}\, U(2^n)$ is a finite subset and $\mu{:} \, \D\,{\to}\, \RR_+$ is a probability distribution on $\D$.  The pair $(\D,\mu)$ is  called a unitary $2$-design~\cite{dankert2009exact} if
\be
\label{2design_def1}
\sum_{\hat{U}\in \D} \mu(\hat{U})
(\hat{U}^\dag \hat{A} \hat{U})\otimes (\hat{U}^\dag \hat{B} \hat{U})
=\int_{U(2^n)}  (\hat{U}^\dag \hat{A} \hat{U})\otimes (\hat{U}^\dag \hat{B} \hat{U}) dU
\ee
for any complex matrices $\hat{A}$ and $\hat{B}$. Here the tensor product separates two $n$-qubit registers and the integral in the right-hand side of Eq.~(\ref{2design_def1}) is the average over the Haar distribution on the unitary group $U(2^n)$.  We reserve the hat notation for complex unitary matrices to avoid confusion with binary symplectic matrices considered in the rest of the paper.  Below we choose $\D$ to be the $n$-qubit Clifford group and construct a probability distribution $\mu$ that minimizes the average cost
\be
\label{cost(mu)}
\sum_{\hat{U}\in \D} \mu(\hat{U}) \cdot \mathrm{cost}(\hat{U}),
\ee
subject to the constraint that $(\D,\mu)$ is a unitary $2$-design.  Here $\mathrm{cost}(\hat{U})$ is the minimum number of the $\cnotgate$ gates required to implement $\hat{U}$ by a quantum circuit composed of the Hadamard, Phase, and $\cnotgate$ gates. 

Since Pauli operators have zero cost, we can assume wlog that the optimal solution $\mu$ is Pauli-invariant, i.e., $\mu(\hat{U})\,{=}\,\mu(\hat{U}\hat{O})$ for all $n$-qubit Pauli operators $\hat{O}$.  As discussed in \sec{def}, the unitary version of the $n$-qubit Clifford group is isomorphic to $\C_n\,{\times}\, \{I,X,Y,Z\}^n$. Here we ignore the overall phase factors.  Define the probability distribution $\pi{:} \, \C_n\,{\to}\, \RR_+$ such that $\pi(U)\,{=}\,4^n \mu(U{\times} P)$ for all $U\,{\in}\, \C_n$ and $P\,{\in}\, \{I,X,Y,Z\}^n$.  The distribution $\pi$ is well-defined whenever $\mu$ is Pauli-invariant. In Appendix~\ref{app:B} we show that $\mu$ is a Clifford $2$-design iff $\pi$ obeys the so-called Pauli mixing constraint~\cite{cleve2015near}
\be
\label{pauli_mixing}
\mathrm{Pr}_{U\sim \pi}[Ux{=}y]:= 
\longsum[23]_{U\in \C_n : \, Ux=y} \; \pi(U) = \frac1{4^n-1} \quad \mbox{for all non-zero vectors $x,y \in \{0,1\}^{2n}$}.
\ee 
Furthermore, $\mu$ has the average cost
\be
\label{cost(pi)}
\sum_{U\in \C_n} \pi(U) \cdot \mathrm{cost}(U).
\ee
Thus it suffices to minimize the average cost Eq.~(\ref{cost(pi)}) over variables $\pi(U)\,{\ge}\, 0$ subject to the normalization constraint $\sum_{U\in \C_n} \pi(U)\,{=}\,1$ and the Pauli mixing constraint, Eq.~(\ref{pauli_mixing}).  This gives a linear program with $|\C_n|$ variables.

The next step is to reduce the number of variables and the number of constraints in the linear program.  Suppose $\pi$ is a Pauli mixing distribution on $\C_n$, that is, $\pi$ obeys Eq.~(\ref{pauli_mixing}). Define a symmetrized version of $\pi$ as follows. First, sample $U\,{\in}\, \C_n$ from the distribution $\pi$.  Second, sample $W\,{\in}\,S_n$ and $L,R\,{\in}\, \C_n^0$ from the uniform distribution on the respective groups. Finally, output $U'\,{=}\,LW^{-1}U WR$.  The probability distribution of $U'$ is given by
\[
\pi'(U')=\frac1{6^{2n} n!} \longsum[8]_{L,R\in \C_n^0} \; \sum_{W\in S_n} \pi(WL^{-1} U'R^{-1} W^{-1}).
\]
Since the cost is invariant under a qubit relabeling and left/right multiplications by the elements of local subgroup $\C_n^0$, the distributions $\pi$ and $\pi'$ have the same average cost.  We claim that $\pi'$ is Pauli mixing.  Indeed, pick any non-zero vectors $x,y\,{\in}\, \{0,1\}^{2n}$, a qubit permutation $W\,{\in}\, S_n$, and local Cliffords $L,R\,{\in}\, \C_n^0$.  Then
\be
\label{pauli_mixing1}
\mathrm{Pr}_{U\sim \pi}[LW^{-1}UWRx\,{=}\,y]=\mathrm{Pr}_{U\sim \pi}[Ux'\,{=}\,y'] = \frac1{4^n-1},
\ee
where $x'\,{=}\,WRx\,{\ne}\, 0$ and $y'\,{=}\,WL^{-1} y\,{\ne}\, 0$. The last equality in Eq.~(\ref{pauli_mixing1}) follows from the assumption that $\pi$ is Pauli mixing.  Thus $\pi'$ is a convex linear combination of Pauli mixing distributions, that is, $\pi'$ itself is Pauli mixing.
 
The above shows that an optimal Clifford $2$-design can be found by minimizing the average cost Eq.~(\ref{cost(pi)}) over symmetric Pauli mixing distributions $\pi$  such that the probability $\pi(U)$ depends  only on the equivalence class $[U]$ that contains $U$. Such distribution $\pi$ can be compactly specified by considering the set of reduced elements
\[
\R_n:= \{ \mathsf{ReduceU}(U): \, U\,{\in}\, \C_n\}.
\] 
Given a reduced element $U\,{\in}\, \R_n$, define the probability distribution
\[
\eta(U) = \longsum[6]_{U'\in [U]} \pi(U') =\pi(U)\cdot | [U] |.
\]
Note that $\eta$ is a probability distribution on $\R_n$
since each equivalence class $[U]$ contains a unique reduced element, see \ssec{reduceu}.  For brevity, we will refer to $\eta$ as a reduced distribution.  The average cost of the original distribution $\pi$ depends only on $\eta$ and can be computed using the formula
\be
\label{cost(eta)}
\longsum[5]_{U\in \R_n} \eta(U) \cdot \mathrm{cost}(U).
\ee
It remains to express the Pauli mixing constraint in terms of the reduced distribution $\eta$. Given a reduced element $U\,{\in}\, \R_n$ and non-zero vectors $x,y\,{\in}\, \{0,1\}^{2n}$, define the quantity
\[
g(U,x,y) = \frac{\#\{U'\in [U]\, : \, U'x=y\}}{|[U]|}. 
\]
In words, $g(U,x,y)$ is the probability that a random uniformly distributed element of the equivalence class $[U]$ maps $x$ to $y$.  Then $\pi$ is Pauli mixing iff
\be
\label{pauli_mixing2}
\longsum[5]_{U\in \R_n} \eta(U) g(U,x,y)= \frac1{4^n-1}
\ee
for all non-zero vectors $x,y \in \{0,1\}^{2n}$.  It remains to note that some constraints Eq.~(\ref{pauli_mixing2}) are redundant.  Indeed, since the equivalence class $[U]$ is invariant under the left/right multiplications of $U$ by the elements of the local subgroup $\C_n^0$, one has $g(U,x,y)\,{=}\,g(U,Lx,Ry)$ for all $L,R\in \C_n^0$.  Suppose $(x_j,x_{n+j})\,{\ne}\, (0,0)$ for some qubit $j$.  Then one can choose $L\,{\in}\, \C_n^0$ acting non-trivially only on the $j$-th qubit such that $(Lx)_j\,{=}\,0$ and $(Lx)_{n+j}\,{=}\,1$,  see \ssec{reduceu}.  Applying this transformation to all qubits we conclude that the Pauli mixing constraint Eq.~(\ref{pauli_mixing2}) has to be imposed only for vectors
\be
\label{pauli_mixing3}
x,y\in \{(0^n z) : \, z\in \{0,1\}^n{\setminus} 0^n\}.
\ee
Minimizing the average cost Eq.~(\ref{cost(eta)}) over variables $\eta(U)\,{\ge}\, 0$ with $U\,{\in}\, \R_n$, subject to the normalization $\sum_{U\in \R_n} \eta(U)\,{=}\,1$ and the Pauli mixing constraints Eqs.~(\ref{pauli_mixing2},\ref{pauli_mixing3}), gives a linear program with $|\R_n|$ variables and $1\,{+}\,(2^n{-}1)^2$ equality constraints.  We were able to find an optimal solution of this linear program numerically for $n\,{=}\,2,3,4$ qubits.  The optimal reduced distributions $\eta$ presented in \tab{n2d}, \tab{n3d}, and \tab{n4d} are compactly represented by a list of reduced elements $U_1,U_2,\ldots,U_m \in \R_n$ along with their probabilities $\eta(U_j)$.  Only reduced elements that appear with non-zero probability are shown.  The tables display an optimal circuit implementation of each reduced element $U_j$. To avoid clutter, we omit single-qubit gates on the left and on the right.  The actual $2$-design has the form $LW^{-1}U_jW R$, where the index  $j\in \{1,2,\ldots,m\}$ is sampled with the probability $\eta(U_j)$, the qubit permutation $W$ is sampled uniformly from $S_n$, and $L,R$ are sampled uniformly from the local subgroup $C_n^0$.

\begin{table}
\captionsetup{width=17cm}
\begin{center}
\begin{tabular}{r|c}
\hline
circuit $U_j$ & probability $\eta(U_j)$  \\
\hline
\rule{0pt}{4ex} 
$
\Qcircuit @C=.7em @R=.7em {
& \ctrl{1} & \qw \\
& \targ &  \qw \\
}
$ 
& $0.6$ \rule[-5ex]{0pt}{4ex}\\
\hline
\rule{0pt}{4ex} 
$
\Qcircuit @C=.7em @R=.7em {
& \ctrl{1} & \targ & \qw \\
& \targ & \ctrl{-1} & \qw \\
}
$ 
& $0.3$ \rule[-5ex]{0pt}{4ex}\\
\hline
\rule{0pt}{4ex} 
$
\Qcircuit @C=.7em @R=.7em {
& \ctrl{1} & \targ & \ctrl{1} & \qw \\
& \targ & \ctrl{-1} & \targ & \qw \\
}
$ 
& $0.1$ \rule[-5ex]{0pt}{4ex}\\
\hline
\end{tabular}
\caption{Optimal two-qubit Clifford $2$-design with
the average cost $1.5$. This coincides with the average 
cost of the full Clifford group $\C_2$.}
\label{tab:n2d}
\end{center}
\end{table}

\begin{table}
\captionsetup{width=17cm}
\begin{center}
\begin{tabular}{r|c||r|c}
\hline
circuit $U_j$ & probability $\eta(U_j)$ & circuit $U_j$ & probability $\eta(U_j)$   \\
\hline
\rule{0pt}{4ex} 
$
\Qcircuit @C=.7em @R=.7em @!{
& \targ &\ctrl{1} & \qw \\
& \ctrl{-1} & \targ &  \qw \\
& \qw & \qw & \qw \\
}
$ 
& $0.074175$ &
$
\Qcircuit @C=.7em @R=.7em @!{
& \targ & \ctrl{2} & \ctrl{1} & \targ & \qw \\
& \qw & \qw & \targ & \ctrl{-1} & \qw \\
& \ctrl{-2} & \targ & \qw & \qw & \qw \\
}
$
& $0.098901$  \rule[-9ex]{0pt}{4ex}
\\
\hline
\rule{0pt}{4ex}
$
\Qcircuit @C=.7em @R=.7em @!{
& \targ & \ctrl{1} & \targ & \qw \\
& \ctrl{-1} & \targ & \ctrl{-1} & \qw \\
& \qw & \qw & \qw & \qw \\
}
$ 
& $0.035715$ &
$
\Qcircuit @C=.7em @R=.7em @!{
& \targ & \ctrl{1} & \ctrl{2} & \targ & \qw \\
& \ctrl{-1} & \targ & \qw & \ctrl{-1} & \qw \\
& \qw & \qw & \targ & \qw & \qw \\
}
$
& $0.098901$  \rule[-9ex]{0pt}{4ex}
\\
\hline
\rule{0pt}{4ex}
$
\Qcircuit @C=.7em @R=.7em @!{
& \targ & \qw & \ctrl{1} & \qw \\
& \qw  & \ctrl{1} & \targ & \qw \\
& \ctrl{-2} & \targ & \qw & \qw \\
}
$ 
& $0.692309$ 
& &  \rule[-9ex]{0pt}{4ex}
\\
\hline
\end{tabular}
\caption{Optimal three-qubit Clifford $2$-design with
the average cost $3.12363...$. For comparison, the full
Clifford group $\C_3$ has the average 
cost $3.50937...$.}
\label{tab:n3d}
\end{center}
\end{table}

\begin{table}
\captionsetup{width=17cm}
\begin{center}
\begin{tabular}{r|c||r|c}
\hline
circuit $U_j$ & $\eta(U_j)$ & circuit $U_j$ & $\eta(U_j)$   \\
\hline
\rule{0pt}{4ex} 
$
\Qcircuit @C=.4em @R=.4em @!{
& \targ & \ctrl{2} & \gate{\pgate} & \targ & \ctrl{1} & \gate{\pgate} &\targ & \qw \\
& \ctrl{-1} & \qw & \qw & \qw & \targ & \qw & \qw & \qw \\
& \qw & \targ & \qw & \qw & \qw & \qw & \ctrl{-2} & \qw \\
& \qw & \qw & \qw &  \ctrl{-3} & \qw & \qw & \qw & \qw
}
$ 
& $0.141176$ & 
$
\Qcircuit @C=.7em @R=.7em @!{
&  \ctrl{1} & \targ & \ctrl{1} & \qw \\
& \targ & \ctrl{-1} & \targ & \qw \\
& \ctrl{1} & \targ & \qw & \qw \\
& \targ & \ctrl{-1} & \qw & \qw \\
}
$
& $0.023663$ 
\rule[-13ex]{0pt}{4ex} \\
\hline
\rule{0pt}{4ex} 
$
\Qcircuit @C=.4em @R=.4em @!{
& \ctrl{1} & \targ & \gate{\pgate} & \qw & \qw & \targ & \qw \\
& \targ & \qw & \gate{\pgate} & \targ & \qw & \ctrl{-1} & \qw \\
& \qw & \qw & \qw & \ctrl{-1} & \gate{\hgate} & \ctrl{1} & \qw \\
& \qw & \ctrl{-3} & \qw & \qw & \qw & \targ & \qw
}
$
& $0.009893$ & 
$
\Qcircuit @C=.4em @R=.4em @!{
& \targ & \ctrl{1} & \gate{\pgate} & \gate{\hgate} & \qw & \ctrl{2} & \qw \\
& \qw & \targ & \qw & \ctrl{1} & \qw & \qw & \qw \\
& \qw & \ctrl{1} & \qw & \targ & \gate{\pgate} & \targ & \qw \\
& \ctrl{-3} & \targ & \qw & \qw & \qw & \qw & \qw \\
}
$
& $0.013368$ 
\rule[-14ex]{0pt}{4ex} \\
\hline
\rule{0pt}{4ex} 
$
\Qcircuit @C=.4em @R=.4em @!{
& \qw & \targ & \ctrl{1} & \gate{\pgate} & \gate{\hgate} & \qw & \ctrl{2} & \qw \\
& \qw & \qw & \targ & \qw & \ctrl{1} & \qw & \qw & \qw \\
& \ctrl{1} & \qw & \qw & \qw & \targ & \gate{\pgate} & \targ & \qw \\
& \targ & \ctrl{-3} & \qw & \qw & \qw & \qw & \qw & \qw
}
$ 
& $0.146526$ & 
$
\Qcircuit @C=.4em @R=.4em @!{
& \ctrl{2} & \qw & \targ & \gate{\pgate} & \gate{\hgate} & \qw & \ctrl{1} & \qw \\
& \qw & \ctrl{2} & \qw & \qw & \targ & \gate{\pgate} & \targ & \qw \\
& \targ & \qw & \qw & \qw & \ctrl{-1} & \qw & \qw & \qw \\
& \qw & \targ & \ctrl{-3} & \qw & \qw & \qw & \qw & \qw
}
$
& $0.014572$ 
\rule[-13ex]{0pt}{4ex} \\
\hline
\rule{0pt}{4ex} 
$
\Qcircuit @C=.4em @R=.4em @!{
& \ctrl{1} & \gate{\hgate} & \ctrl{3} & \gate{\pgate} & \qw & \targ & \qw & \qw  \\
& \targ & \qw & \qw & \gate{\pgate} & \ctrl{1} & \qw & \targ & \qw  \\
& \qw & \qw & \qw & \qw & \targ & \ctrl{-2} & \qw & \qw \\
& \qw & \qw & \targ & \qw & \qw & \qw & \ctrl{-2} & \qw 
}
$
& $0.164572$  & 
$
\Qcircuit @C=.7em @R=.7em @!{
& \ctrl{1} & \qw & \targ & \qw & \ctrl{3} & \qw \\
& \targ & \ctrl{1} & \qw & \qw & \qw & \qw \\
& \qw & \targ & \ctrl{-2} & \targ & \qw & \qw \\
& \qw & \qw & \qw & \ctrl{-1} & \targ & \qw 
}
$
& $0.206952$ 
\rule[-13ex]{0pt}{4ex} \\
\hline
\rule{0pt}{4ex} 
$
\Qcircuit @C=.4em @R=.4em @!{
& \targ & \qw & \qw & \qw & \qw & \qw & \ctrl{3} & \qw & \qw \\
& \qw & \ctrl{1} & \qw & \qw & \qw & \qw & \qw & \targ & \qw \\
& \ctrl{-2} & \targ & \gate{\pgate} & \gate{\hgate} & \ctrl{1} & \qw & \qw & \qw & \qw \\
& \qw & \qw & \qw & \qw & \targ & \gate{\pgate} & \targ & \ctrl{-2} & \qw 
}
$
& $0.198930$  & 
$
\Qcircuit @C=.7em @R=.7em @!{
& \ctrl{1} & \targ & \ctrl{1} & \qw \\
& \targ & \ctrl{-1} & \targ & \qw \\
& \ctrl{1} & \targ & \ctrl{1} & \qw \\
& \targ & \ctrl{-1} & \targ & \qw 
}
$
&  $0.007353$ 
\rule[-13ex]{0pt}{4ex} \\
\hline
\rule{0pt}{4ex} 
$
\Qcircuit @C=.7em @R=.7em @!{
& \ctrl{3} & \targ & \qw & \qw & \ctrl{2} & \qw & \qw \\
& \qw & \qw & \ctrl{1} & \qw & \qw & \targ & \qw \\
& \qw & \qw & \targ & \ctrl{1} & \targ & \ctrl{-1} & \qw \\
& \targ & \ctrl{-3} & \qw & \targ & \qw & \qw & \qw  
}
$ 
& $0.072994$  &  & 
\rule[-13ex]{0pt}{4ex} \\
\hline
\end{tabular}
\caption{Optimal four-qubit Clifford $2$-design with the average cost $5.08034...$. For comparison, the full Clifford group $\C_4$ has the average cost $5.85856...$. We note that all except for two
circuits in the above table have cost $5$. The remaining pair of circuits have cost $6$.}
\label{tab:n4d}
\end{center}
\end{table}

\subsection{Comparison to prior work}

Similar-spirited prior work includes the synthesis of 4-qubit optimal Clifford circuits \cite{kliuchnikov2013optimization}, the synthesis of 4-bit optimal reversible circuits \cite{golubitsky2011study}, and optimal solution of Rubik's cube puzzle \cite{rokicki2014diameter}.  \cite{kliuchnikov2013optimization} is most closely related to our work, given the focus on Clifford circuits; the difference is we chose to study the two-qubit gate cost, which better reflects the constraints of the existing quantum computers than the total gate count.  The search space size comparison is $4.7 {\cdot} 10^{10}$ in \cite{kliuchnikov2013optimization} to $2.1 {\cdot} 10^{23}$ in our work---an almost $13$ orders of magnitude difference.  \cite{golubitsky2011study} study reversible circuits, being a highly relevant type of computations. Their search space size is $2.1 {\cdot} 10^{13}$, meaning we solved a problem with $10$ orders of magnitude higher search space size.  Finally, \cite{rokicki2014diameter} studies Rubik's cube, which is also a finite group.  Their search space size is $4.3{\cdot} 10^{19}$, meaning ours is almost $4$ orders of magnitude higher.

\section{Conclusion}
In this paper, we reported algorithms and their C\texttt{++} implementation that compute all two-qubit gate count optimal 6-qubit Clifford circuits.  There are about $2.1 {\cdot} 10^{23}$ different Clifford functions.  The large search space required us to employ server-class machines to make the computation possible.  In particular, we used HPC to break down the set of canonical representatives of Clifford group elements sharing similar optimal circuit structure, and store them in a database of size $2.1$TB.  Given this database on an SSD and a $2.5$GB index file in RAM, the time to extract an optimal circuit using a consumer-grade laptop is $\laptophottime$ seconds---$10$ times faster than the typical access time for a spindle drive.  The time to extract an optimal circuit using an enterprise-level system while storing the database in RAM is $\serverhottime$ seconds---$15$ times faster than the typical HDD access time.  We used the database to establish the maximal gate count needed to implement an arbitrary 6-qubit Clifford unitary and showed the distribution of the number of Clifford functions across their required gate counts.  We established a new example of quantum advantage by Clifford circuits over $\cnotgate$ gate circuits and found optimal Clifford 2-designs for the number of qubits up to, and including, $4$.

\section*{Data availability}
A Python implementation of the described algorithms will be available at:\\ \url{https://github.com/qiskit-community/prototype-clifford-optimizer}.

\bibliographystyle{unsrt}

\appendix
\section{Proof of \lem{gen1}}\label{appe:A}
We need to show that any element $U\,{\in}\,\C_n^k$ can be written as $U=G_{a_1} G_{a_2} \cdots G_{a_k} L$ for some $L\,{\in}\,\C_n^0$ and some $k$-tuple of generators.  We use the induction in $k$.  The base of induction is $k{=}0$, in which case the statement is trivial. Suppose $k{\ge}1$ and  $U\,{\in}\,\C_n^k$.  By definition, $U$ can be implemented by a circuit composed of $k$ $\cnotgate$ gates and some number of single-qubit gates.  Let $\cnotgate_{i,j}$ be the last $\cnotgate$ gate in this circuit. Then 
\[
U=M \cnotgate_{i,j} V
\]
for some $M\,{\in}\, \C_n^0$ and $V\,{\in}\, \C_n^{k-1}$.  
We can assume without loss of generality that $i<j$.
Indeed, if $i>j$, use the identity $\cnotgate_{j,i}=\hgate_i \hgate_j \cnotgate_{i,j} \hgate_i \hgate_j$ to flip the control and the target qubits of the last $\cnotgate$ gate. The extra $\hgate$ gates can be absorbed into $M$ and $V$ layers.
By the induction hypothesis, $V=G_{a_2} \cdots G_{a_k} L$ for some $L\,{\in}\, \C_n^0$.  Furthermore, we can assume without loss of generality that $M\,{=}\,\Gate{a}_i \Gate{b}_j$ for some $\Gate{a},\Gate{b}\,{\in}\, \C_1$.  Indeed, all single-qubit gates in $M$ that act on qubits $\ell \,{\notin}\, \{i,j\}$ can be commuted through $\cnotgate_{i,j}$ and absorbed into $V$. If $\Gate{a},\Gate{b}\,{\in}\, \{\idgate,\hgate\pgate,\pgate\hgate\}$, we are done.  Indeed, in this case $\Gate{a}_i \Gate{b}_j\cnotgate_{i,j}\,{=}\,G_{a_1}$ is a generator and $U=G_{a_1} V = G_{a_1}G_{a_2} \cdots G_{a_k} L$ with $L\,{\in}\, \C_n^0$. Otherwise, transform  $\Gate{a}$ and $\Gate{b}$ into the desired form by ``borrowing" the missing single-qubit gates from $V$ and commuting them through $\cnotgate_{i,j}$ using the  Clifford group identities\footnote{Recall that these identities only apply to elements of the binary symplectic group; the corresponding identities for unitary Clifford operators may include some extra phase factors and Pauli gates.}
: 
$$\pgate^2=\hgate^2=(\pgate\hgate\pgate)^2=\idgate, \quad \pgate\hgate\pgate=\hgate\pgate\hgate, $$
$$
\pgate_i \cnotgate_{i,j} = \cnotgate_{i,j} \pgate_i, \quad
\hgate_i \cnotgate_{i,j} = (\hgate\pgate)_i \cnotgate_{i,j} \pgate_i, \quad
(\pgate\hgate\pgate)_i \cnotgate_{i,j} = (\pgate\hgate)_i \cnotgate_{i,j} \pgate_i,$$
$$
\pgate_j \cnotgate_{i,j} = (\hgate\pgate)_j\cnotgate_{i,j} (\pgate\hgate\pgate)_j, 
\hgate_j \cnotgate_{i,j} = (\pgate\hgate)_j \cnotgate_{i,j} (\pgate\hgate\pgate)_j, \text{ and }
(\pgate\hgate\pgate)_j \cnotgate_{i,j} = \cnotgate_{i,j} (\pgate\hgate\pgate)_j. 
$$
This completes the proof.

\section{Pauli mixing constraint}
\label{app:B}

In this section we prove that a Pauli-invariant probability distribution $\mu$ on the $n$-qubit Clifford group is a unitary $2$-design iff $\mu$ is Pauli mixing.  The fact that Pauli-invariance and Pauli mixing are sufficient for being a $2$-design is known \cite[Appendix D]{cleve2015near}.  Thus it suffices to prove that any Pauli-invariant Clifford $2$-design is Pauli mixing.

The Haar integeral in Eq.~(\ref{2design_def1}) can be computed explicitly using Weingarten functions~\cite{collins2006integration},
\[
\int_{U(2^n)}  (\hat{U}^\dag \hat{A} \hat{U})\otimes (\hat{U}^\dag \hat{B} \hat{U}) dU
=\mathrm{SWAP} \left[ \frac{\mathrm{Tr}(\hat{A}\hat{B})}{4^n-1} - \frac{\mathrm{Tr}(\hat{A})\mathrm{Tr}(\hat{B})}{2^n(4^n-1)}\right]
+ \hat{I}\otimes \hat{I} \left[ 
\frac{\mathrm{Tr}(\hat{A})\mathrm{Tr}(\hat{B})}{4^n-1} - \frac{\mathrm{Tr}(\hat{A}\hat{B})}{2^n(4^n-1)}\right].
\]
Here $\mathrm{SWAP}$ is a unitary operator that swaps the two $n$-qubit registers separated by the tensor product.  It is well-known that any complex matrix of size $2^n{\times} 2^n$ can be expanded in the Pauli basis
\[
\PL_n=\{\hat{I},\hat{X},\hat{Y},\hat{Z}\}^{\otimes n}.
\]
Thus it suffices
to impose  Eq.~(\ref{2design_def1}) only for $\hat{A},\hat{B}\in \PL_n$.
Noting that the Pauli basis is orthonormal with respect to the inner product
$\mathrm{Tr}(\hat{A}^\dag \hat{B})/2^n$ one concludes that a pair $(\D,\mu)$ is 
a unitary $2$-design iff
\be
\label{2design_def2}
\sum_{\hat{U}\in \D} \mu(\hat{U})
(\hat{U}^\dag \hat{A} \hat{U})\otimes (\hat{U}^\dag \hat{B} \hat{U})
= \left\{ \begin{array}{rcl}
0 & \mbox{if} & \hat{A}\ne \hat{B}, \\
\hat{\Lambda}  & \mbox{if} &  \hat{A}= \hat{B}\ne \hat{I}\\
\end{array}\right. \quad \mbox{for all $\hat{A},\hat{B}\in  \PL_n$}
\ee
where
\[
\hat{\Lambda} = 
\frac1{4^n-1} \;(2^n \mathrm{SWAP} -  \hat{I}\,{\otimes}\, \hat{I})=
\frac1{4^n-1} \; \longsum[15]_{\hat{O}\in \PL_n\setminus \{\hat{I}\}} \; \hat{O}\otimes \hat{O}.
\]

A Pauli operator $\hat{O}\,{\in}\, \PL_n$ can be parameterized by a bit string $v\,{\in}\, \{0,1\}^{2n}$ such that
\[
\hat{O}(v) \equiv \hat{O}(v_1v_{n+1})\otimes \hat{O}(v_2 v_{n+2}) \otimes \cdots \otimes \hat{O}(v_n v_{2n}),
\]
where $\hat{O}(00)\,{\equiv}\, \hat{I}$, $\hat{O}(10)\,{\equiv}\, \hat{X}$, $\hat{O}(01)\,{\equiv}\, \hat{Z}$, and $\hat{O}(11)\,{\equiv}\, \hat{Y}$.  The unitary version of the Clifford group, which we denote $\CU_n$, is a group of complex matrices $\hat{U}\,{\in}\,U(2^n)$ that map Pauli operators to Pauli operators under conjugation. More formally, $\hat{U}\in \CU_n$ iff there exists a symplectic matrix $U\,{\in}\,\C_n$ such that
\be
\label{action_on_pauli}
\hat{U}\hat{O}(v)\hat{U}^\dag =\pm \hat{O}(Uv)
\ee
for all $v\,{\in}\, \{0,1\}^{2n}$. Here the sign may depend on $v$.  The symplectic matrix $U\,{\in}\, \C_n$ in Eq.~(\ref{action_on_pauli}) is uniquely determined by $\hat{U}$.  Conversely, $\hat{U}$ is uniquely determined by $U$ up to (right) multiplications by Pauli operators and the overall phase. In other words, $\CU_n$ is isomorphic (as a set) to $\C_n {\times} \PL_n$ if one ignores the overall phase of unitary matrices.

Suppose $\mu{:} \, \CU_n\,{\to}\, \RR_+$ is a Pauli-invariant probability distribution, that is, $\mu(\hat{U})\,{=}\,\mu(\hat{U}\hat{O})$ for all $\hat{O}\,{\in}\, \PL_n$ and $\hat{U}\,{\in}\, \CU_n$.  Using the isomorphism $\CU_n\,{\cong}\, \C_n {\times} \PL_n$, define a distribution $\pi{:} \, \C_n\,{\to}\, \RR_+$ such that $\mu(U{\times}P)=\pi(U)/4^n$ for all $U\,{\in}\, C_n$ and $P\,{\in}\, \PL_n$.  Suppose $(\CU_n,\mu)$ is a  $2$-design, that is, $\mu$ obeys Eq.~(\ref{2design_def2}) with ${\cal D}=\CU_n$.  Consider the second case of Eq.~(\ref{2design_def2}) such that $\hat{A}\,{=}\,\hat{B}\,{=}\,\hat{O}(x)$ for some non-zero vector $x\,{\in}\, \{0,1\}^{2n}$.  Then it is equivalent to
\[
\sum_{U\in \C_n} \pi(U)
\hat{O}(Ux) \otimes \hat{O}(Ux)
= 
\frac1{4^n-1} \; \longsum[26]_{y \in \{0,1\}^{2n}\setminus 0^{2n}} \; \hat{O}(y)\otimes \hat{O}(y).
\]
Since Pauli operators are linearly independent, this is possible only if a random vector $Ux$ with $U$ sampled from $\pi(U)$ is distributed uniformly on the set of all non-zero vectors $\{0,1\}^{2n}{\setminus} 0^{2n}$.  This gives the Pauli mixing condition Eq.~(\ref{pauli_mixing}).

\end{document}